\documentclass[onecolumn, 12pt]{IEEEtran}

\usepackage{cite, color}
\usepackage{amsmath, amsfonts, amsthm, amssymb, esint}
\usepackage{epsfig}
\usepackage{subfigure}
\usepackage{comment}
\usepackage{algorithm, algpseudocode}
\newcommand{\PR}[1]{\mathrm{Pr}[#1]}
\newcommand{\NN}{\nonumber}
\newcommand{\NNL}{\nonumber\\}

\newtheorem{theorem}{Theorem}
\newtheorem{lemma}{Lemma}
\newtheorem{corollary}{Corollary}

\newtheorem{example}{Example}
\newtheorem{observation}{Observation}

\theoremstyle{definition}
\newtheorem{definition}{Definition}

\title{Optimal Feedback Rate Sharing Strategy in Zero-Forcing MIMO Broadcast Channels}
\author{\IEEEauthorblockN{Jung~Hoon~Lee},~\IEEEmembership{Student Member,~IEEE},~and
    \IEEEauthorblockN{Wan~Choi},~\IEEEmembership{Senior Member,~IEEE}
\thanks{Parts of this paper have been presented in IEEE
    International Conference on ITS Telecommunication (ITST), Lille,
    France, Oct. 2009. }% <-this % stops a space
\thanks{J.~H.~Lee and W.~Choi are with Department of Electrical
    Engineering, Korea Advanced Institute of Science and Technology
    (KAIST), Daejeon 305-701, Korea (e-mail: tantheta@kaist.ac.kr,
    wchoi@ee.kaist.ac.kr). }% <-this % stops a space
}

\begin{document}
\maketitle

\begin{abstract}
In this paper, we consider a multiple-input multiple-output
broadcast channel with limited feedback where all users share the
feedback rates. Firstly, we find the optimal feedback rate sharing
strategy using zero-forcing transmission scheme at the transmitter
and random vector quantization at each user. We mathematically prove
that equal sharing of sum feedback size among all users is the
optimal strategy in the low signal-to-noise ratio (SNR) region,
while allocating whole feedback size to a single user is the optimal
strategy in the high SNR region. For  the mid-SNR region, we propose
a simple numerical method to find the optimal feedback rate sharing
strategy based on our analysis and show that the equal allocation of
sum feedback rate to a partial number of users is the optimal
strategy. It is also shown that the proposed simple numerical method
can be applicable to finding the optimal feedback rate sharing
strategy when different path losses of the users are taken into
account.
We show that our proposed feedback rate sharing scheme can be
extended to the system with stream control and is still useful for
the systems with other techniques such as regularized zero-forcing
and spherical cap codebook.

\end{abstract}


\IEEEpeerreviewmaketitle

%%%%%%%%%%%%%%%%%%%%%%%%%%%%%%%%%%%%%%%%%%%%%%%%%%%%%%%%%%%%%%%%%%%
% % % % % % % % % % % % % % % % % % % % % % % % % % % % % % % % % %
%%%%%%%%%%%%%%%%%%%%%%%%%%%%%%%%%%%%%%%%%%%%%%%%%%%%%%%%%%%%%%%%%%%
\section{Introduction}

In recent years, multiple-input multiple-output (MIMO) broadcast
channel (BC) systems, constructed by an access point with multiple
antennas and many users, have been intensively studied \cite{CS2003,
VT2003, WSS2006}. In a MIMO BC, multiple users are simultaneously
served through independent user specific multiple data streams and a
\emph{multiplexing gain} is attained as in point-to-point MIMO. The
capacity region of the Gaussian MIMO BC was derived in
\cite{WSS2006} where dirty paper coding (DPC) \cite{C1983} is known
to be a capacity achieving scheme. Because DPC is hard to implement,
many practical techniques have been proposed such as zero-forcing
precoding (channel inversion) \cite{YG2006} and Tomlinson-Harashima
precoding \cite{ZSE2002}. In these schemes, multiuser interference
is pre-canceled at the transmitter with perfect channel state
information at the transmitter (CSIT).

CSIT can be obtained by reciprocity between uplink and downlink
channels in time division duplexing (TDD) systems and feedback from
receivers in frequency division duplexing (FDD) systems. In FDD
systems, the amount of feedback information is in general limited
and hence perfect CSIT is not available. The accuracy of CSIT
depends on both the type of feedback technique and the amount of
feedback overhead allowed.
A popular feedback architecture is a codebook approach where an
index of a codeword in a predetermined codebook is fed back to the
transmitter \cite{LHSH2003}. There have been many studies on the
performance of codebook based multi-user MIMO systems using various
transmission schemes such as zero-forcing (ZF) beamforming
\cite{J2006}, block diagonalization (BD) \cite{RJ2008,CLL2010}, and
the unitary precoding \cite{KPLK2009}.

In limited feedback environments, a key difference between  MIMO BC
and point-to-point MIMO is the multiplexing gain achievability
\cite{LHSH2003, J2006}.
%
%\cite{LH2005_1, HL2005, LHSH2003}.
%
In point-to-point MIMO, a full multiplexing gain is achievable even
with open-loop transmission. On the other hand, a full multiplexing
gain cannot be achieved using a finite amount of feedback
information in a MIMO BC \cite{J2006}. The multiplexing gain of MIMO
BC rather diminishes in the high signal-to-noise ratio (SNR) region
due to imperfect orthogonalization resulting from inaccurate CSIT.
To maintain the multiplexing gain, it was shown in \cite{J2006,
RJ2008} that the feedback size should linearly increase with SNR (in
decibel scale).

Since a large amount of feedback is a heavy burden on uplink
capacity, many studies have been devoted to increasing the
efficiency of limited feedback. In \cite{J2008}, a feedback
reduction technique has been proposed using multiple antennas at the
receiver. User selection in MIMO BC has been studied to reduce the
amount of uplink feedback \cite{SH2005, YJG2007, CFAH2007, APL2009,
HR2013}.
In \cite{YJG2007}, random beamforming was generalized and
semi-orthogonal user selection was proposed. Also, it was shown that
channel quality information as well as channel direction information
 are necessary to obtain both the maximum multiplexing and
diversity gains. In \cite{APL2009}, a dual-mode limited feedback
system was proposed to switch between single user and multiuser
transmissions. The authors in \cite{HR2013} investigated two partial
feedback schemes for user scheduling.

In practical systems, the uplink capacity of control channels is
typically limited and shared among multiple users.
A sum feedback rate constraint in space division multiple access
(SDMA) was considered in \cite{HHA2007} but the amount of feedback
information per user was held constant. In \cite{RJ2009}, the
optimum feedback size per user and the number of feedback users were
investigated under a sum feedback rate constraint assuming all users
employ the same amount of feedback. Recently, strategies of feedback
bit partitioning between the desired and interfering channels
proposed in \cite{BH2011} for a cooperative multicell system.
In $K$-user multiple-input-single-output (MISO) interference
channel, the feedback rate control to minimize the average
interference power was proposed in \cite{HLK2011}.

In MIMO BC, the effects of different amounts of feedback size among
the users are studied in \cite{CKCK2008, LC2009, XZD2009, KY2012}.
In \cite{CKCK2008}, the feedback rate sharing strategy has been
proposed to minimize the upper bound of sum rate loss in correlated
single-polarized and dual-polarized channels, respectively. The
feedback rate sharing strategies in the low and high SNR regions
have been proposed in terms of the correlation coefficient.
The feedback rate sharing strategy to increase the sum rate was also
proposed in \cite{XZD2009} by considering users' path losses, where
the system performance was shown to be improved by changing feedback
bit allocation according to the path losses. However, when the path
losses are similar, the feedback rate sharing strategy in
\cite{XZD2009} is to equally share the sum feedback size regardless
of SNR levels but it is not optimal in some SNR regions. Also, the
effects of path losses are canceled out in the high SNR region so
that equal sharing of the sum feedback size is not optimal any more.
The feedback rate sharing strategy to minimize total transmission
power for given users' outage probabilities was proposed in
\cite{KY2012}.

In this paper, we provide a new analytical framework for the
feedback rate sharing strategy and rigorously analyzed the effects
of different amounts of feedback information among users by
extending and generalizing the results of \cite{LC2009}. The effects
of feedback rate sharing on the achievable rate are investigated in
a MIMO BC with ZF beamforming at the transmitter and random vector
quantization (RVQ) \cite{SH2005_2} at each user. We derive the
optimal feedback rate sharing strategies according to various SNR
regions.
Our analytical results prove the optimal feedback rate sharing
strategy in the low and the high SNR regions. The feedback rate
should be equally shared among all users in the low SNR region while
the whole feedback rate should be allocated to a single user in the
high SNR region.
For the mid-SNR region, we establish a simple numerical method for
finding the optimal feedback sharing strategy based on our
analytical framework. Through the proposed numerical method, we find
that to equally allocate whole feedback size to a partial number of
users is the optimal feedback rate sharing strategy.
For the users suffering different path losses, we show that the
proposed numerical method can be applicable to finding the optimal
feedback rate sharing strategy. In the high SNR region, we prove
that the effects of path losses are canceled out and hence the
optimal feedback strategy is to allocate the whole feedback size to
a single user with the highest SNR.
Our proposed feedback rate sharing strategy derived from the system
with ZF beamforming and RVQ is also evaluated for the systems with
other techniques such as stream control, regularized ZF transmission
scheme and spherical cap codebook model \cite{MSEA2003, YJG2007}.
Our numerical results show that our proposed feedback rate sharing
strategy is still valid for other configurations.

The rest of this paper is organized as follows. We describe the
system model and formulate the problem in Section II.
The impacts of asymmetric feedback size among users are investigated
in Section III. The optimal sum feedback rate sharing strategy is
derived in Section IV. The numerical results are shown in Section V.
Section VI concludes our paper.

%%%%%%%%%%%%%%%%%%%%%%%%%%%%%%%%%%%%%%%%%%%%%%%%%%%%%%%%%%%%%%%%%%%
% % % % % % % % % % % % % % % % % % % % % % % % % % % % % % % % % %
%%%%%%%%%%%%%%%%%%%%%%%%%%%%%%%%%%%%%%%%%%%%%%%%%%%%%%%%%%%%%%%%%%%
\section{Problem Formulation}
\subsection{System Model}

Our system model is depicted in Fig. \ref{fig1}. We consider a MIMO
BC with $M$ transmit antennas and $K (=M)$ users having a single
antenna.
If the receiver has multiple antennas, each antenna can be
considered as an independent user, or receive combining discussed in
\cite{J2008} can be adopted.
The received signal at the user $k$ becomes
\begin{equation}
    y_k=\sqrt{\gamma_k}\mathbf{h}_k^\dagger\mathbf{x}+n_k, \quad
    k=1,\ldots,K, \NN
\end{equation}
where $\gamma_k$ is the path loss of the $k$th user, $\mathbf{h}_k
\in \mathbb{C}^{M\times 1}$ is a channel vector whose entries are
independent and identically distributed (i.i.d.) circularly
symmetric complex Gaussian random variables with zero mean and unit
variance, $\mathbf{x} \in \mathbb{C}^{M\times 1}$ is the transmit
signal vector, $n_k$ is a complex Gaussian noise with zero mean and
unit variance, and the superscript $\dagger$ denotes conjugate
transposition of a vector.
When $P$ is the transmit signal power, $\mathbf{x}$ satisfies that
$\mathbb{E} [tr\left(\mathbf{x} \mathbf{x}^\dagger\right)]=P$.
If users demand the same quality of service, the propagation path
losses need to be pre-compensated to yield the same average SNR at
the receiver in downlink. Thus, we firstly assume that the different
propagation path losses for users are compensated by the
transmitter, i.e., $\gamma_1 = \gamma_2 = \cdots = \gamma_K=1$. The
open loop power control is also useful for preventing waste of
transmit power and avoiding extra interference to other users.
Then, we extend our results to different path loss scenarios in
Section \ref{sec:different_paths}.

As a simple linear precoding scheme, we adopt a ZF beamforming
scheme in which the data stream for each user is aligned with its
precoding vector. We denote the precoding vector of the $k$th user
as $\mathbf{v}_k$ such that $\Vert\mathbf{v}_k \Vert=1$ and then the
transmit signal $\mathbf{x}$ becomes $\mathbf{x}=\sum_{k=1}^{K}
\mathbf{v}_k s_k$,
where $s_k$ is the data symbol for the $k$th user.
We assume that the transmitter has only channel direction
information (CDI) so that the feedback for power allocation can be
saved. Therefore, the transmitter allocates equal power to users
such that $\mathbb{E}|s_k|^2=P/M$. Also, we assume that $s_k$ is
chosen from a Gaussian codebook and the codeword block length is
sufficiently long so that it encounters all possible channel
realizations for ergodicity.
Obviously, power adaptation can further increase the achievable rate
but the power allocation using channel quality information (CQI) is
a secondary problem when the number of transmit antennas is same as
the number of served users, i.e., full multiplexing \cite{J2006}.
In Section \ref{section:stream_control}, we will consider the stream
control where the transmitter adaptively controls multiplexing gain
and the served users equally share total transmit power.

The received signal at the $k$th user using linear precoding becomes
\begin{align}
    y_k=\mathbf{h}_k^\dagger\mathbf{v}_k s_k
        + \sum_{i=1,i\ne k}^{K} \mathbf{h}_k^\dagger\mathbf{v}_i s_i +n_k, \quad
        k=1,\ldots,K. \label{eqn:y_k}
\end{align}
When the transmitter knows $\{\mathbf{h}_1, \ldots, \mathbf{h}_K\}$
perfectly, the precoding vectors yield zero multiuser interferences,
i.e., $\sum_{i\ne k}\mathbf{h}_k^\dagger\mathbf{v}_i s_i =0$;
the received signal at the $k$th user becomes
\begin{align}
    y_k=\mathbf{h}_k^\dagger\mathbf{v}_k s_k + n_k, \quad k=1,\ldots,K.
    \NN
\end{align}
In most practical systems, however, the imperfect CSI is only
available at the transmitter due to the limited feedback budget.
The user $k$ quantizes its own channel, $\mathbf{h}_k$, and feeds
the quantized CSI denoted by $\hat{\mathbf{h}}_k$ to the
transmitter.
Then, the transmitter finds the precoding vectors $\mathbf{v}_1,
\ldots, \mathbf{v}_K$ from the quantized CSI, $\hat{\mathbf{h}}_1,
\ldots, \hat{\mathbf{h}}_K$, instead of the perfect CSI,
$\mathbf{h}_1, \ldots, \mathbf{h}_K$.
Because of the quantization errors, the precoding vectors obtained
from the quantized CSI cannot perfectly mitigate the multiuser
interference.
The precoding vector cannot be exactly picked in the null space of
the other users' channel vectors; the interference term $\sum_{i\ne
k} \mathbf{h}_k^\dagger \mathbf{v}_i s_i$ remains in the received
signal.

At the transmitter, a quantized channel matrix defined by
$\hat{\mathbf{H}} \triangleq [\hat{\mathbf{h}}_1,\ldots,
\hat{\mathbf{h}}_K] ^\dagger$ is constructed with the quantized CSI
fed back from the users. The $k$th normalized column vector of
$\hat{\mathbf{H}}^{-1}$ becomes the precoding vector for the $k$th
user, $\mathbf{v}_k$, where $(\cdot)^{-1}$ denotes the matrix
inversion.
Thus, we can decompose $\hat{\mathbf{H}}^{-1}$ as
$\hat{\mathbf{H}}^{-1}=\mathbf{V}\mathbf{\Lambda}$,
%
%\begin{align}
%    \hat{\mathbf{H}}^{-1}=\mathbf{V}\mathbf{\Lambda},
%\end{align}
%
where $\mathbf{V}=\left[\mathbf{v}_1, \ldots,\mathbf{v}_K\right]$ is
a zero-forcing beamforming matrix such as $\Vert\mathbf{v}_k
\Vert^2=1$, and $\mathbf{\Lambda} =
\textrm{diag}(\lambda_1,\ldots,\lambda_K)$ is diagonal matrix whose
element $\lambda_k \in \mathbb{R}^+$ is the Euclidean norm of the
$k$th column of $\hat{\mathbf{H}}^{-1}$.

For the channel quantization, RVQ is considered at each user, which
is widely used to analyze the effects of quantization error and
asymptotically optimal as the number of antennas goes to infinity
\cite{AL2007, J2006}.
Although the performance is suboptimal for a small feedback size,
RVQ makes the analysis tractable and provides insightful results.
Furthermore, the overall trends of RVQ generally agree with the
trends of other quantization models \cite{YJG2007}.

Using $b_k$-bit RVQ at the $k$th user, the quantized CSI is obtained
by
\begin{align}
\hat{\mathbf{h}}_k
    =\underset{\mathbf{w}\in \mathcal{W}_k}{\arg\max}
        ~~\cos^2(\angle(\mathbf{h}_k, \mathbf{w}))%\\
    =\underset{\mathbf{w}\in \mathcal{W}_k}{\arg\max}
        ~~|\mathbf{h}_k^\dagger\mathbf{w}|^2, \NN
\end{align}
where $\mathcal{W}_k=\{\mathbf{w}_{k,1}, \ldots, \mathbf{w}
_{k,2^{b_k}}\}$ is a random vector codebook at the $k$th user
consists of $2^{b_k}$ randomly chosen isotropic $M$-dimensional unit
vectors.
The quantization error denoted by $Z_k \in [0,1]$ becomes
\begin{align}
Z_k=\underset{\mathbf{w}\in \mathcal{W}_k}{\min}
        \sin^2 (\angle(\mathbf{h}_k,\mathbf{w}))\label{eqn:Z_k}
    =\sin^2 (\angle(\mathbf{h}_k,\hat{\mathbf{h}}_k))
    =1-\vert \tilde{\mathbf{h}}_k^\dagger \hat{\mathbf{h}}_k\vert^2,
\end{align}
where $\tilde{\mathbf{h}}_k = \mathbf{h}_k/\Vert\mathbf{h}_k \Vert$.
For an arbitrary codeword $\mathbf{w}\in \mathcal{W}_k$, $\vert
\tilde{\mathbf{h}}_k ^\dagger \mathbf{w}\vert^2$ is a squared inner
product of two independent random vectors isotropic in
$\mathbb{C}^M$, so follows the beta distribution%
\footnote{The probability density function of beta distributed
random variable $S$ with parameters ($a,b$) becomes
$f_S(s)=\frac{\Gamma(a+b)}{\Gamma(a)\Gamma(b)}s^{a-1}(1-s)^{b-1}$
\cite[p.635]{Z2003}.}
with parameters $(M-1,1)$ \cite{J2006, AL2007}. Consequently, a
quantization error using $b_k$-bit RVQ, $Z_k$, becomes the minimum
of $2^{b_k}$ independent beta distributed random variables with
parameters $(M-1,1)$.
Correspondingly the complementary cumulative density function (CDF)
of $Z_k$ is given by \cite{AL2007}
\begin{align}
   \PR{Z_k>z}=\left(1-z^{M-1}\right)^{2^{b_k}}.
    \label{eqn:QE_CDF}
\end{align}

%%%%%%%%%%%%%%%%%%%%%%%%%%%%%%%%%%%%%%%%%%%%%%%%%%%%%%%%%%%%%%%%%%%
% % % % % % % % % % % % % % % % % % % % % % % % % % % % % % % % % %
%%%%%%%%%%%%%%%%%%%%%%%%%%%%%%%%%%%%%%%%%%%%%%%%%%%%%%%%%%%%%%%%%%%
\subsection{Feedback Rate Sharing Strategy}
We assume an \emph{average} feedback size allocated for each user is
$\bar{b}$ so that the total feedback rate (i.e., the sum of all
individual users' feedback rates) becomes $K\bar{b}$ bits per
channel realization.
Assuming the feedback rate sharing among users, each user uses
$b_k$-bit feedback and the sum feedback rate constraint becomes
$\sum_{k=1}^K b_k = K\bar{b}$.
Since codebook size is typically a non-negative integer number of
bits, we restrict the average feedback size, $\bar{b}$, as an
positive integer, i.e., $\bar{b}\in \mathbb{Z}^+$.
For the same reason, we assume the feedback size at the $k$th user,
$b_k$, as a non-negative integer, i.e., $b_k\in \{0\} \cup
\mathbb{Z}^+$ for $k=1, \ldots, K$,

From individual feedback rates, a feedback rate sharing strategy can
be expressed by $K$-dimensional vector
\begin{align}
    \mathbf{b}=[b_1, \ldots, b_K],
\end{align}
and the sum feedback rate constraint becomes $\Vert \mathbf{b}
\Vert_1 = K\bar{b}$ where $\Vert\cdot\Vert_1$ is the vector one
norm.

From \eqref{eqn:y_k}, we obtain the average sum rate as a function
of transmit power, $P$, and the sum feedback rate sharing strategy,
$\mathbf{b}$, denoted by $\mathcal{R}(P,\mathbf{b})$ given by
\begin{align}
    \mathcal{R}(P, \mathbf{b})
    = \sum_{k=1}^{K}\mathbb{E}\left[
    \log_2\left(
    1+\frac{\frac{P}{M} |\mathbf{h}_k^\dagger \mathbf{v}_k|^2}
    {1+\sum_{i\ne k}\frac{P}{M} |\mathbf{h}_k^\dagger \mathbf{v}_i|^2}
    \right) \right].
    \label{eqn:sum_rate}
\end{align}
Thus, we solve the following problem:
\begin{align}
    \underset{\mathbf{b}=[b_1, \ldots, b_K]}{\textrm{maximize}}
        &\qquad \mathcal{R}(P, \mathbf{b}) \label{eqn:optimization_problem}\\
    \textrm{subject to}
        &\qquad \sum_{k=1}^K b_k = K\bar{b},\label{eqn:constraint1}\\
        &\qquad b_k\in \{0\} \cup \mathbb{Z}^+ \quad k=1,\ldots, K.
        \label{eqn:constraint2}
\end{align}
Note that the optimal sum feedback rate sharing strategy will be
derived later and shown to be dependent on the SNR value. Therefore,
the feedback bits are reallocated each time when the SNR changes.
In practical scenarios, several allocation patterns can be
constructed offline for typical SNR values and then the transmitter
can broadcast an appropriate allocation pattern using the current
SNR.

%%%%%%%%%%%%%%%%%%%%%%%%%%%%%%%%%%%%%%%%%%%%%%%%%%%%%%%%%%%%%%%%%%%
% % % % % % % % % % % % % % % % % % % % % % % % % % % % % % % % % %
%%%%%%%%%%%%%%%%%%%%%%%%%%%%%%%%%%%%%%%%%%%%%%%%%%%%%%%%%%%%%%%%%%%
\section{Impacts of Asymmetric Feedback Sizes among Users}

To find the optimal feedback rate sharing strategy, we first analyze
the impact of asymmetric feedback sizes among the users on the sum
rate.
For the simplicity, we define three random variables
\begin{align}
    Q_k \triangleq \Vert\mathbf{h}_k\Vert^2,\quad
    X_k \triangleq |\tilde{\mathbf{h}}_k^\dagger \mathbf{v}_k|^2,\quad
    Y_k \triangleq \sum_{i\ne k}|\tilde{\mathbf{h}}_k^\dagger
    \mathbf{v}_i|^2,
    \label{eqn:three_RV}
\end{align}
where
$Q_k$ is the $k$th channel gain, $X_k$ is the squared inner product
between the $k$th normalized channel vector and the $k$th
beamforming vector, and $Y_k$ is the sum of the squared inner
products between the $k$th normalized channel vector and the other
beamforming vectors.
Note that $X_k$ is not affected by the feedback size of the $k$th
user since $\mathbf{v}_k$ is selected in the null space of
$\{\hat{\mathbf{h}}_i\}_{i\neq k}$.

Using the quantization error $Z_k$ defined in \eqref{eqn:Z_k}, we
can decompose $\tilde{\mathbf{h}}_k$ into $\tilde{\mathbf{h}}_k =
\sqrt{1-Z_k} \hat{\mathbf{h}}_k + \sqrt{Z_k}\mathbf{e}_k$ where
$\mathbf{e}_k$ is an unit vector such that $|\hat{\mathbf{h}}_k
^\dagger\mathbf{e}_k|^2=0$.
The random variable $Y_k$ becomes
\begin{align}
 Y_k%&=\sum_{i\ne k}|\tilde{\mathbf{h}}_k^\dagger \mathbf{v}_i|^2\\
    &=\sum_{i\ne k}\left|\left(\sqrt{1-Z_k}\hat{\mathbf{h}}_k+\sqrt{Z_k}
        \mathbf{e}_k\right)^\dagger\mathbf{v}_i\right|^2 \label{eqn:quantization}\\
    &=Z_k\sum_{i\ne k}|\mathbf{e}_k^\dagger\mathbf{v}_i|^2 \label{eqn:orthogonality1}\\
    &=Z_k\cdot W_k,
\end{align}
where the random variable $W_k \triangleq \sum_{i\ne k}|
\mathbf{e}_k ^\dagger \mathbf{v}_i|^2$ is the sum of the square of
inner products between the quantization error vector $\mathbf{e}_k$
and the beamforming vectors of other users $\{\mathbf{v}_i\}_{i\ne
k}$. The independency between $Z_k$ and $|\mathbf{e} ^{\dagger}_k
\mathbf{v}_i|^2$ is shown in [12] from the fact that the magnitude
of the quantization error, $Z_k$ is independent of the direction of
quantization error, $\mathbf{e}_k$. Thus, we can easily find that
$Z_k$ and $W_k (=\sum_{i\ne k}|\mathbf{e} ^{\dagger}_k
\mathbf{v}_i|^2)$ are independent.
We start from the following lemma.
%
%\vspace{0.1in}

%%%%%%%%%%%%%%%%%%%%%%%%%%%%%%%%%%%%%%%%%%%%%%%%%%%%%%%%%%%%%%%%%%%
% % % % % % % % % % % % % % % % % % % % % % % % % % % % % % % % % %
%%%%%%%%%%%%%%%%%%%%%%%%%%%%%%%%%%%%%%%%%%%%%%%%%%%%%%%%%%%%%%%%%%%
\begin{lemma}\label{lemma:RV_properties}
The random variables $Q_k$, $X_k$, $W_k$ and $Z_k$ have following
properties.
\begin{enumerate}
\item Invariant with the feedback sizes, $b_1, \ldots, b_K$, the
distributions of $Q_k$, $X_k$, and $W_k$ are identical for all
users, respectively, i.e.,
\begin{align}
    f_{Q_k}(q) = f_{Q_1}(q),\quad
    f_{X_k}(x) = f_{X_1}(x),\quad \NNL
    f_{W_k}(w) =f_{W_1}(w), \quad  k = 2, \ldots, K,\NN
\end{align}
where $f_{Q_k}(q)$, $f_{X_k}(x)$, and $f_{W_k}(w)$ are the marginal
PDFs of $Q_k$, $X_k$, $W_k$, respectively,

\item $Q_k$, $X_k$, and $W_k$ are independent of $Z_k$, respectively.

\item The joint PDF of $Q_k$, $X_k$, and $W_k$ are identical for all
users, i.e.,
\begin{align}
 f_{Q_k, X_k, W_k}(q, x, w) = f_{Q_1, X_1, W_1}(q, x, w), \NN
%    \quad  k = 2, \ldots, K
\end{align}
where $f_{Q_k,X_k,W_k}(q,x,w)$ is the joint PDF of $Q_k$, $X_k$, and
$W_k$.

\end{enumerate}

\end{lemma}
%%%%%%%%%%%%%%%%%%%%%%%%%%%%%%%%%%%%%%%%%%%%%%%%%%%%%%%%%%%%%%%%%%%
\begin{proof}
See Appendix A. %\ref{appendix:RV_properties}.
\end{proof}
%%%%%%%%%%%%%%%%%%%%%%%%%%%%%%%%%%%%%%%%%%%%%%%%%%%%%%%%%%%%%%%%%%%
% % % % % % % % % % % % % % % % % % % % % % % % % % % % % % % % % %
%%%%%%%%%%%%%%%%%%%%%%%%%%%%%%%%%%%%%%%%%%%%%%%%%%%%%%%%%%%%%%%%%%%
%Also, we obtain the following lemma.
%%%%%%%%%%%%%%%%%%%%%%%%%%%%%%%%%%%%%%%%%%%%%%%%%%%%%%%%%%%%%%%%%%%
% % % % % % % % % % % % % % % % % % % % % % % % % % % % % % % % % %
%%%%%%%%%%%%%%%%%%%%%%%%%%%%%%%%%%%%%%%%%%%%%%%%%%%%%%%%%%%%%%%%%%%
\begin{lemma}\label{lemma:capacity_k}
The achievable rate of the $k$th user is determined by only its own
feedback size $b_k$ and is independent of the other users' feedback
sizes $ \{b_i\}_{i\ne k}$.
\end{lemma}
%%%%%%%%%%%%%%%%%%%%%%%%%%%%%%%%%%%%%%%%%%%%%%%%%%%%%%%%%%%%%%%%%%%
\begin{proof}
From Lemma \ref{lemma:RV_properties}, we can rewrite the average sum
rate in \eqref{eqn:sum_rate} as
\begin{align}
    \mathcal{R}(P, \mathbf{b})
    &= \sum_{k=1}^{K}\mathbb{E}_{Q_k,X_k,W_k,Z_k}\left[
        \log_2\left(
        1+\frac{\frac{P}{M} Q_kX_k}{1+ \frac{P}{M} Q_k W_k Z_k}
        \right) \right]\NNL
    &= \sum_{k=1}^{K}\mathbb{E}_{Q_1,X_1,W_1,Z_k}\left[
        \log_2\left(
        1+\frac{\frac{P}{M} Q_1X_1}{1+ \frac{P}{M} Q_1 W_1 Z_k}
        \right) \right].\NN
\end{align}
Thus, the achievable rate at the $k$th user is dependent on only its
own feedback size because $Q_1$, $X_1$, and $W_1$ are not affected
by the feedback size as noted in Lemma \ref{lemma:RV_properties}.
Since the distribution of $Z_k$ is a function of $b_k$, the
achievable rate at each user is only affected by its own feedback
size.
\end{proof}
%%%%%%%%%%%%%%%%%%%%%%%%%%%%%%%%%%%%%%%%%%%%%%%%%%%%%%%%%%%%%%%%%%%
% % % % % % % % % % % % % % % % % % % % % % % % % % % % % % % % % %
%%%%%%%%%%%%%%%%%%%%%%%%%%%%%%%%%%%%%%%%%%%%%%%%%%%%%%%%%%%%%%%%%%%

Thus, the achievable rate of the user $k$ becomes a function of
transmit power $P$ and own feedback size $b_k$ denoted by
$\mathcal{R}_k(P, b_k)$ such that
\begin{align}
 \mathcal{R}_k(P, b_k)
    = \mathbb{E}_{Q_1,X_1,W_1,Z_k}\left[\log_2\left(
    1+\frac{\frac{P}{M} Q_1X_1}{1+ \frac{P}{M} Q_1 W_1 Z_k}
    \right) \right], \label{eqn:capacity_k}
\end{align}
and it satisfies that $\mathcal{R}(P, \mathbf{b}) = \sum_{k=1}^K
\mathcal{R}_k(P, b_k)$.

To verify Lemma \ref{lemma:capacity_k}, two feedback scenarios
$\mathbf{b}_1=[10,10,10]$ and $\mathbf{b}_2=[10, 0, 0]$ are
considered in ZF MIMO BC with $M=3$, $K=3$.
In Fig. \ref{fig:capacity_user1},
the sum rate for the first scenario is much higher than that for the
second scenario due to the larger amount of total feedback
information. As predicted in Lemma \ref{lemma:capacity_k}, however,
the achievable rate of user 1 is the same in the two scenarios.

Lemma \ref{lemma:capacity_k} indicates that a feedback size of a
user does not affect the achievable rates of the other users and
only changes its own achievable rate.
Under a sum feedback rate constraint, an increase of one user's
feedback size necessarily decreases other users' feedback sizes.
With more accurate $\hat{\mathbf{h}}_k$, the transmitter can pick
the beamforming vectors of other users in more accurate null space
of the user $k$. Hence, the user $k$ benefits from less interference
from other users. On the other hand, the other users experience more
interference since the accuracy of the users' channel knowledge
degrades under the sum feedback rate constraint.
Consequently, when a user increases its own feedback size, the
achievable rate of the user increases but the achievable rates of
the other users decrease, and vice versa. The optimal feedback rate
sharing strategy starts from this fundamental tradeoff.

%%%%%%%%%%%%%%%%%%%%%%%%%%%%%%%%%%%%%%%%%%%%%%%%%%%%%%%%%%%%%%%%%%%
% % % % % % % % % % % % % % % % % % % % % % % % % % % % % % % % % %
%%%%%%%%%%%%%%%%%%%%%%%%%%%%%%%%%%%%%%%%%%%%%%%%%%%%%%%%%%%%%%%%%%%
\section{Sum Feedback Rate Sharing Strategy}
\subsection{Low SNR Region}

In the low SNR region, the achievable rate of the $k$th user given
in \eqref{eqn:capacity_k} becomes
\begin{align}
  &\lim_{P\to 0}\mathcal{R}_k(P, b_k)\NNL
  &= \lim_{P\to 0}\mathbb{E}\bigg[\log_2\left(1+\frac{P}{M} Q_1 X_1
        + \frac{P}{M} Q_1 W_1  Z_k\right)\NNL
        &\qquad- \log_2\left(1+ \frac{P}{M} Q_1 W_1  Z_k\right)\bigg] \NNL
  &=\lim_{P\to 0}\mathbb{E}\bigg[\log_2\left(1+\frac{P}{M} Q_1 X_1\right)
    + \log_2\left(1+\frac{\frac{P}{M}Q_1W_1Z_k}{1+\frac{P}{M}Q_1X_1}\right)\NNL
    &\qquad- \log_2\left(1+\frac{P}{M}Q_1W_1Z_k\right) \bigg] \NNL
%
%  &\stackrel{(a)}{=} \frac{1}{\ln 2}\mathbb{E}\left[\frac{P}{M} Q_1 X_1\right]
%    +\frac{1}{\ln 2}\mathbb{E}\left[
%    \frac{P}{M}Q_1W_1Z_k\cdot\left(\frac{1}{1+\frac{P}{M}Q_1X_1}-1\right)\right]\NNL
%
  &\stackrel{(a)}{=} \frac{1}{\ln 2}\mathbb{E}\left[\frac{P}{M} Q_1 X_1\right]
    -\frac{1}{\ln 2}\mathbb{E}\left[\frac{\frac{P^2}{M^2}Q_1^2X_1W_1Z_k}{1+\frac{P}{M}Q_1X_1}\right]
    \NNL
  &\stackrel{(b)}{=} \frac{1}{\ln 2}\mathbb{E}\left[\frac{P}{M} Q_1 X_1\right]
    -\frac{1}{\ln 2}\mathbb{E}\left[\frac{\frac{P^2}{M^2}Q_1^2X_1W_1}{1+\frac{P}{M}Q_1X_1}\right]
    \cdot \mathbb{E}[Z_k],\NN
\end{align}
where the equality $(a)$ holds because  $\lim_{x \to 0} \ln(1+x) =
x$, and the equality $(b)$ holds from the fact that $Z_k$ is
independent of $Q_k$, $X_k$, and $W_k$ from Lemma
\ref{lemma:RV_properties}.
In the low SNR region, therefore, the optimization problem
\eqref{eqn:optimization_problem} is equivalent with the following
problem:
\begin{align}
 \underset{\mathbf{b}=[b_1, \ldots, b_K]}{\textrm{minimize}} & \qquad \sum_{k=1}^K \mathbb{E}[Z_k]
 \label{eqn:optimization_problem_L}\\
    \textrm{subject to}
        &\qquad \eqref{eqn:constraint1},\eqref{eqn:constraint2}. \NN
\end{align}

\begin{definition}[Majorization]
For a vector $\mathbf{a}\in\mathbb{R}^m$, we denote by
$\mathbf{a}^\downarrow \in \mathbb{R}^m$ the vector with the same
components, but sorted in decreasing order. For given vectors
$\mathbf{a}_1, \mathbf{a}_2\in \mathbb{R}^m$ such that
$\Vert\mathbf{a}_1\Vert_1 = \Vert\mathbf{a}_2\Vert_1$, we say
$\mathbf{a}_1$ majorizes $\mathbf{a}_2$ written as
$\mathbf{a}_1\succeq\mathbf{a}_2$ when
\begin{align}
    \sum_{i=1}^n [\mathbf{a}_1^\downarrow]_i &\ge \sum_{i=1}^n [\mathbf{a}_2^\downarrow]_i
        \qquad 1\le n \le m,
\end{align}
where $[\cdot]_i$ denotes the $i$th component of the vector.
\end{definition}

\begin{theorem}[Strategy in the Low SNR Region]\label{theorem:strategy_low}
Using RVQ in the low SNR region, feedback rate sharing strategy
$\mathbf{b}_1$ achieves higher average sum rate than feedback rate
sharing strategy $\mathbf{b}_2$ whenever $\mathbf{b}_1 \preceq
\mathbf{b}_2$, i.e.,
\begin{align}
 \lim_{P\to 0} \mathcal{R}(P,\mathbf{b}_1)
    \ge \lim_{P\to 0} \mathcal{R}(P,\mathbf{b}_2) \quad\textrm{for all}\quad
    \mathbf{b}_1 \preceq \mathbf{b}_2.
\end{align}
\end{theorem}
%
%%%%%%%%%%%%%%%%%%%%%%%%%%%%%%%%%%%%%%%%%%%%%%%%%%%%%%%%%%%%%%%%%%%
\begin{proof}
See Appendix B. %\ref{appendix:strategy_low}.
\end{proof}
%%%%%%%%%%%%%%%%%%%%%%%%%%%%%%%%%%%%%%%%%%%%%%%%%%%%%%%%%%%%%%%%%%%
% % % % % % % % % % % % % % % % % % % % % % % % % % % % % % % % % %
%%%%%%%%%%%%%%%%%%%%%%%%%%%%%%%%%%%%%%%%%%%%%%%%%%%%%%%%%%%%%%%%%%%
\begin{corollary}\label{corollary:strategy_low}
In the low SNR region, when the sum feedback rates is $K\bar{b}$
(i.e., $\sum b_k=K\bar{b}$), the optimal feedback rate sharing
strategy is to allocate the same amount of feedback ($b_k=\bar{b}$)
to all users while the worst strategy is to allocate whole feedback
amount $K\bar{b}$ to a single user.
\end{corollary}
\begin{proof}
All possible feedback sharing strategies $\mathbf{b}$ ($\Vert
\mathbf{b} \Vert_1 = K\bar{b}$) satisfy that
\begin{align}
 [\bar{b}, \ldots, \bar{b}] \preceq \mathbf{b} \preceq [K\bar{b}, 0, \ldots, 0].
\end{align}
Thus, the optimal feedback sharing strategy in low SNR region is to
allocate the same feedback size to all users while the worst
strategy is to allocate the whole feedback size to a single user.
\end{proof}

%%%%%%%%%%%%%%%%%%%%%%%%%%%%%%%%%%%%%%%%%%%%%%%%%%%%%%%%%%%%%%%%%%%
\subsection{High SNR Region}\label{section:high_SNR_region}
With fixed feedback size in the high SNR region, the sum rate of a
MIMO BC saturates and cannot achieve the full multiplexing gain
\cite{J2006}. This is because the remaining interference caused by
the quantization error increases with SNR so that the SINR is
saturated in the high SNR region.

For ease of explanation, we decompose the achievable rate at user
$k$ into an \emph{increasing term} and a \emph{decreasing term}
denoted by $\mathcal{R}_k^+ (P,b_k)$ and $\mathcal{R}_k^- (P,b_k)$,
respectively, given by
\begin{align}
 \mathcal{R}_k^+(P,b_k)&=
    \mathbb{E}\left[ \log_2 \left( 1+\frac{P}{M} Q_1 X_1
    + \frac{P}{M} Q_1 W_1 Z_k \right)\right]\NNL
%    \label{eqn:capacity_gain}\\
 \mathcal{R}_k^-(P,b_k)&=
    \mathbb{E}
    \left[\log_2 \left( 1+ \frac{P}{M} Q_1 W_1  Z_k\right)
    \right],\NN
%    \label{eqn:capacity_loss}
\end{align}
so that $\mathcal{R}_k(P,b_k) = \mathcal{R}_k^+(P,b_k) -
\mathcal{R}_k^-(P,b_k)$.
Similarly, we can express the average sum rate into two parts as
$\mathcal{R}(P,\mathbf{b}) = \mathcal{R}^+(P,\mathbf{b}) -
\mathcal{R}^-(P,\mathbf{b})$
%
%\begin{align}
% \mathcal{R}(P,\mathbf{b})
%    = \mathcal{R}^+(P,\mathbf{b}) - \mathcal{R}^-(P,\mathbf{b}),
%\end{align}
%
where $\mathcal{R}^+(P,\mathbf{b}) =
\sum_{k=1}^K\mathcal{R}_k^+(P,b_k)$ and
$\mathcal{R}^-(P,\mathbf{b})=\sum_{k=1}^K \mathcal{R}_k^-(P,b_k)$.

In the high SNR region, the increasing term of the $k$th user's
achievable rate, $\mathcal{R}_k^+(P,b_k)$, becomes
\begin{align}
 &\lim_{P \to \infty} \mathcal{R}_k^+(P,b_k)
    =\mathbb{E}\left[ \log_2 \left(\tfrac{P}{M} Q_1\right)\right]
    + \mathbb{E}\left[ \log_2 \left(X_1 + W_1 Z_k \right)\right],
    \NN
\end{align}
where the second term on the right hand side of the equality is only
affected by the quantization error, $Z_k$.
For the quantization error $Z_k \in [0,1]$, the range of $\log_2
\left(X_1 + W_1 Z_k \right)$ becomes
$\log_2 \left(X_1 + W_1 Z_k \right) \in [\log_2 \left(X_1\right),
\log_2 \left(X_1 + W_1 \right)]$.
In the high SNR region, on the other hand, the decreasing term of
the $k$th user's achievable rate, $\mathcal{R}_k^-(P,b_k)$, becomes
\begin{align}
 &\lim_{P \to \infty} \mathcal{R}_k^-(P,b_k)
    =\mathbb{E}\left[ \log_2 \left(\tfrac{P}{M} Q_1W_1\right)\right]
    + \mathbb{E}\left[ \log_2 \left(Z_k \right)\right],\NN
%    \label{eqn:R_gain_low}
\end{align}
where the quantization error affects $\mathbb{E}\left[ \log_2
\left(Z_k \right)\right]$ only.
For the quantization error $Z_k\in[0,1]$, we can find $\log_2
\left(Z_k \right) \in (-\infty, 0]$. However, note that $\log_2
\left(\frac{P}{M}Q_1W_1 \right)\gg-\log_2 Z_k$ when $P \to \infty $
although $\log_2 \left(Z_k \right) \in (-\infty, 0]$.
These facts implicate that in the high SNR region the quantization
error, $Z_k$, only dependent on the feedback size, highly affects
the rate decreasing term $\mathcal{R}_k^-(P,b_k)$ and thus the
achievable rate at each user is dominated by the rate decreasing
term.
Therefore, the feedback rate sharing strategy in the high SNR region
should be focused on minimizing the rate decreasing term.
The average sum rate decreasing term, $\mathcal{R}^-
(P,\mathbf{b})$, becomes
\begin{align}
 \lim_{P\to\infty}\mathcal{R}^-(P,\mathbf{b})
% &= \lim_{P\to\infty}  \sum_{i=1}^K \mathbb{E}
%    \left[\log_2 \left( 1+ \frac{P}{M} Q_1 W_1  Z_k\right)
%    \right].\\
 &=M \mathbb{E}\left[\log_2\left(\frac{P}{M} Q_1 W_1\right)\right]
        + \sum_{k=1}^K \mathbb{E}\left[\log_2 Z_k \right].\NN
\end{align}
Hence, as an alternative of \eqref{eqn:optimization_problem} in the
high SNR region, we solve the optimization problem to minimize
$\mathcal{R}^-(P,\mathbf{b})$ equivalent with the following problem:
\begin{align}
 \underset{\mathbf{b}=[b_1,\ldots,b_K]}{\textrm{minimize}} & \qquad
    \sum_{k=1}^K \mathbb{E}[\log_2 Z_k]
    \label{eqn:optimization_problem_H}\\
    \textrm{subject to}
        &\qquad \eqref{eqn:constraint1},\eqref{eqn:constraint2}. \NN
\end{align}

%%%%%%%%%%%%%%%%%%%%%%%%%%%%%%%%%%%%%%%%%%%%%%%%%%%%%%%%%%%%%%%%%%%
% % % % % % % % % % % % % % % % % % % % % % % % % % % % % % % % % %
%%%%%%%%%%%%%%%%%%%%%%%%%%%%%%%%%%%%%%%%%%%%%%%%%%%%%%%%%%%%%%%%%%%
\begin{theorem}[Strategy in the High SNR Region]\label{theorem:strategy_high}
Using RVQ in the high SNR region, feedback rate sharing strategy
$\mathbf{b}_1$ achieves higher average sum rate than feedback rate
sharing strategy $\mathbf{b}_2$ whenever $\mathbf{b}_1 \succeq
\mathbf{b}_2$, i.e.,
\begin{align}
 \lim_{P\to \infty} \mathcal{R}(P,\mathbf{b}_1)
    \ge \lim_{P\to \infty} \mathcal{R}(P,\mathbf{b}_2) \quad\textrm{for all}\quad
    \mathbf{b}_1 \succeq \mathbf{b}_2.
\end{align}
\end{theorem}
%
%%%%%%%%%%%%%%%%%%%%%%%%%%%%%%%%%%%%%%%%%%%%%%%%%%%%%%%%%%%%%%%%%%%
\begin{proof}
See Appendix C. %\ref{appendix:strategy_high}.
\end{proof}

\begin{corollary}\label{corollary:strategy_high}
In the high SNR region, when the total amount of feedback
information from all users is fixed (i.e., $\sum b_k = K\bar{b}$),
the optimal feedback rate sharing strategy is to allocate whole
feedback amount $K\bar{b}$ to a single user while the worst strategy
is to allocate the same amount of feedback ($b_k=\bar{b}$) to all
users.
\end{corollary}
\begin{proof}
As stated in the proof of Corollary 1, any feedback rate sharing
strategy, $\mathbf{b}$, satisfies that
\begin{align}
 [\bar{b}, \ldots, \bar{b}]
    \preceq \mathbf{b}\preceq [K\bar{b}, 0,\ldots, 0].
\end{align}
Thus, the optimal feedback rate strategy in the high SNR region is
to allocate the whole feedback size to a single user while the worst
strategy is to allocate the same feedback size to each user.
\end{proof}

%%%%%%%%%%%%%%%%%%%%%%%%%%%%%%%%%%%%%%%%%%%%%%%%%%%%%%%%%%%%%%%%%%%
% % % % % % % % % % % % % % % % % % % % % % % % % % % % % % % % % %
%%%%%%%%%%%%%%%%%%%%%%%%%%%%%%%%%%%%%%%%%%%%%%%%%%%%%%%%%%%%%%%%%%%
\subsection{Intermediate SNR Region}\label{section:intermediate}

In Theorem \ref{theorem:strategy_low} and Theorem
\ref{theorem:strategy_high}, the optimal feedback rate sharing
strategies in the asymptotic SNR regions are derived. In the
practical SNR region, the optimal strategy can easily be found by a
numerical method owing to Lemma \ref{lemma:capacity_k} that the
achievable rate of each user only depends on its own feedback size.
We first compute the achievable rates of each user for various
feedback bits $b_k=0, \ldots, K\bar{b}$, respectively. Using the
computed numerical values, we select the best feedback rate sharing
strategy for each SNR that maximizes the total sum rate among all
possible strategies.
For example, when total feedback size is 16bits, the conventional
exhaustive search needs to search the optimal strategy among all
possible 64 strategies. On the other hand, in our proposed numerical
method, it is enough to consider only five strategies -- $[ 0, 0, 0,
0]$, $[ 1, 2, 3, 4]$, $[ 5, 6, 7, 8]$, $[ 9, 10, 11, 12]$, $[ 13,
14, 15, 16]$ -- because the achievable rate for other strategies can
be easily obtained from Lemma \ref{lemma:capacity_k}.
Denoting the set of all possible strategies by $\mathcal{B}$, the
procedure to find the optimal feedback strategy is described in
Algorithm \ref{alg:FRS_procedure}. The complexity of the procedure
will be analyzed in Section \ref{section:complexities}.

\begin{algorithm}[!b]
\caption{Procedure to find Feedback Rate Sharing Strategy}
    \begin{algorithmic}[1]
    \State{\textbf{Initialization:} randomly choose $\mathbf{b}\in\mathcal{B}$}
    \ForAll{$\mathbf{b}'\in \mathcal{B}$}
    \If {$\sum\mathcal{R}_k(\gamma_k P,[\mathbf{b}']_k)
            > \sum\mathcal{R}_k(\gamma_k P,[\mathbf{b}]_k)$}
    \State{$\mathbf{b}=\mathbf{b}'$;}
    \EndIf
    \EndFor
    \State{\textbf{Output:} the optimal feedback strategy $\mathbf{b}$}
    \end{algorithmic}
    \label{alg:FRS_procedure}
\end{algorithm}

\begin{observation}\label{observation:mid_SNR}
The optimal feedback rate sharing strategy is to allocate the same
amount of feedback to the optimal number of users at given SNR.
\end{observation}

\begin{table}[t]\centering
\caption{The optimal feedback rate sharing strategy for a $4\times
4$ MIMO BC}
\begin{tabular}{c|c||c|c||c|c}
 \hline \hline
   \multicolumn{2}{c||}{2 streams} &
   \multicolumn{2}{c||}{3 streams} &
   \multicolumn{2}{c}{4 streams}  \\
 \hline
   SNR{\scriptsize(dB)}& $\mathbf{b}^\downarrow$ & SNR&
   $\mathbf{b}^\downarrow$ &  SNR& $\mathbf{b}^\downarrow$\\
 \hline \hline
     0$\sim$27 & [12,12] & 0$\sim$12 & [8,8,8] & 0$\sim$7 &[6,6,6,6]\\
    28$\sim$~  & [24,~0] & 13$\sim$23 & [12,12,0] & 8$\sim$11 &[8,8,8,0]\\
        ~      &    ~    & 24$\sim$~  & [24,0,0] & 12$\sim$20 &[12,12,0,0]\\
        ~      &    ~    &      ~     &  ~          &   21$\sim$ &[24,0,0,0]\\
 \hline \hline
\end{tabular}
\label{tab:strategy}
\end{table}

\begin{example} \label{example:mid_snr}
For a $4\times 4$ MIMO BC with 24 total allowable feedback bits
($K\bar{b}=24$), the achievable rate of a user for various
$b_k\in\{0,\ldots, 24\}$ is plotted in Fig. \ref{fig:capacity_Bi}.
For various feedback rate sharing strategies, the sum rate is
calculated by using the numerical values obtained in Fig.
\ref{fig:capacity_Bi} and then we can find the optimal feedback
sharing strategy for given SNR as shown in Table \ref{tab:strategy}.
\end{example}
Interestingly, the optimal feedback rate sharing strategy determines
the optimal number of concurrent users for equal feedback rate
sharing at a given SNR.
In a practical system with user scheduling, the weighted sum rate
may be more important than the sum rate. We can also easily find the
optimal feedback rate sharing strategy numerically as in Example 1
owing to Lemma \ref{lemma:capacity_k}.

%%%%%%%%%%%%%%%%%%%%%%%%%%%%%%%%%%%%%%%%%%%%%%%%%%%%%%%%%%%%%%%%%%%
% % % % % % % % % % % % % % % % % % % % % % % % % % % % % % % % % %
%%%%%%%%%%%%%%%%%%%%%%%%%%%%%%%%%%%%%%%%%%%%%%%%%%%%%%%%%%%%%%%%%%%
\subsection{Different Path Losses at the Users} \label{sec:different_paths}
In this subsection, we obtain the feedback rate sharing strategy
according to SNR (i.e., $P$) when propagation path losses for users
are different.
Under the different path losses, the sum rate given in
\eqref{eqn:sum_rate} becomes
\begin{align}
 \mathcal{R}(P, \mathbf{b})
 &= \sum_{k=1}^{K}\mathbb{E}\left[
    \log_2\left(
    1+\frac{\frac{\gamma_k P}{M} |\mathbf{h}_k^\dagger \mathbf{v}_k|^2}
    {1+\sum_{i\ne k}\frac{\gamma_k P}{M} |\mathbf{h}_k^\dagger \mathbf{v}_i|^2}
    \right) \right]\NNL
 &\stackrel{(a)}{=} \sum_{k=1}^{K}\mathbb{E}
    \left[ \log_2\left(
    1+\frac{\frac{\gamma_k P}{M} Q_kX_k}{1+ \frac{\gamma_k P}{M} Q_k W_k Z_k}
    \right) \right]\NNL
 &\stackrel{(b)}{=} \sum_{k=1}^{K}\mathbb{E}
    \left[ \log_2\left(
    1+\frac{\frac{\gamma_k P}{M} Q_1X_1}{1+ \frac{\gamma_k P}{M} Q_1 W_1 Z_k}
    \right) \right]\NN
\end{align}
where $(a)$ is from the definitions of $Z_k$, $Q_k$, $X_k$, and
$W_k$ given in \eqref{eqn:Z_k} and \eqref{eqn:three_RV},
respectively, and $(b)$ holds from Lemma \ref{lemma:RV_properties}.
Thus, we can easily check that Lemma \ref{lemma:capacity_k} is still
valid for different path losses such that $\mathcal{R}(P,
\mathbf{b}) = \sum_{k=1}^K \mathcal{R}_k(\gamma_kP, b_k)$
where $\mathcal{R}_k(\gamma_kP, b_k)$ is the achievable rate at the
$k$th user given by
\begin{align} \label{eqn:pathloss}
 \mathcal{R}_k(\gamma_kP, b_k)
    \triangleq
    \mathbb{E}\left[\log_2\left(
    1+\frac{\frac{\gamma_kP}{M} Q_1X_1}{1+ \frac{\gamma_kP}{M} Q_1 W_1 Z_k}
    \right) \right].
\end{align}
The equation \eqref{eqn:pathloss} indicates that the average
achievable rate at each user is affected by only its own path loss
and independent of other users' path losses. Therefore, the optimal
feedback rate sharing strategy can be found by the simple numerical
method proposed in Section \ref{section:intermediate}.
In the same manner in Example \ref{example:mid_snr}, we first
compute the achievable rates of each user for various feedback bits
based on (\ref{eqn:pathloss}). Then, we select the optimal feedback
rate sharing strategy  $\mathbf{b}=[b_1, \ldots, b_K]$ from the
computed values to maximize the sum rate $\sum_{k=1}^K \mathcal{R}_k
(\gamma_k P, b_k)$.
The equation \eqref{eqn:pathloss} also implicates that the effects
of path losses are canceled out in the high SNR region since
$\lim_{P\to\infty} \mathcal{R}_k (\gamma_k P, b_k) =
\mathbb{E}\left[\log_2\left( 1+\frac{X_1}{W_1 Z_k} \right) \right]$.
Therefore, the optimal feedback rate sharing strategy is the same as
Theorem \ref{theorem:strategy_high} even when different path losses
are taken into account.

On the other hand, the feedback rate sharing strategy for different
path losses proposed in \cite{XZD2009} is given by
\begin{align}
 b_k=\bar{b}-(K-1)\left(
    \log_2 \gamma_k - \frac{1}{K}\sum_{i=1}^K\log_2
    \gamma_i \right) \label{eqn:FRS_Xu}
\end{align}
which results in equal sharing of the sum feedback size regardless
of SNR levels when the path losses are the same (i.e., $\gamma_1 =
\ldots = \gamma_K$), which is not optimal in the mid and the high
SNR regions.

\begin{table}[b]\centering
    \caption{The optimal feedback rate sharing strategy for a
    $4\times 4$ MIMO BC when all users suffering difference path
    losses $(\gamma_1,\gamma_2, \gamma_3,\gamma_4)=(1.5, 1.25, 1,
    0.75)$}
    \begin{tabular}{c|c||c|c}
    \hline
    SNR  & $[b_1,b_2,b_3,b_4]$  & SNR  & $[b_1,b_2,b_3,b_4]$\\\hline
    $~0\sim~1$ dB & $[8,8,8,0]$ & $8\sim17$ dB& $[13,11,0,0]$ \\\hline
    $~2\sim~6$ dB& $[10,8,6,0]$  & $18$ dB & $[16,8,0,0]$ \\\hline
    $7$ dB& $[11,8,5,0]$        & $19$ dB $\sim$& $[24,0,0,0]$ \\\hline
    \end{tabular}\label{tab:strategy_pathlosses}
\end{table}

\begin{example} \label{example:pathloss}
Consider a $4\times 4$ MIMO BC with 24 total allowable feedback bits
($K\bar{b}=24$). We assume the path losses of each user as
$(\gamma_1,\gamma_2, \gamma_3,\gamma_4)=(1.5, 1.25, 1, 0.75)$. For
the given path losses, the feedback rate sharing strategy given in
\eqref{eqn:FRS_Xu} becomes $\mathbf{b}=[7,7,6,4]$. On the other
hand, the optimal feedback rate strategy obtained by the proposed
numerical method is given in Table \ref{tab:strategy_pathlosses}
according to various SNR regions. The average sum rate by the
optimal feedback rate strategy by the proposed method is plotted in
Fig \ref{fig:capacity_pathlosses}.
Fig \ref{fig:capacity_pathlosses} confirms that our proposed
strategy given in Table \ref{tab:strategy_pathlosses} more
significantly outperforms the feedback rate sharing strategy
proposed in \eqref{eqn:FRS_Xu} as SNR becomes higher.
\end{example}

%%%%%%%%%%%%%%%%%%%%%%%%%%%%%%%%%%%%%%%%%%%%%%%%%%%%%%%%%%%%%%%%%%%
% % % % % % % % % % % % % % % % % % % % % % % % % % % % % % % % % %
%%%%%%%%%%%%%%%%%%%%%%%%%%%%%%%%%%%%%%%%%%%%%%%%%%%%%%%%%%%%%%%%%%%
\subsection{Complexity Analysis}\label{section:complexities}

In this subsection, we analyze complexity to find the optimal
feedback rate strategy described in Algorithm
\ref{alg:FRS_procedure}. Because the effects of different path
losses can be simply regarded as different transmit SNR of users as
described in Section IV-D, the achievable rates of users with
different path losses can be calculated by the same procedure based
on Fig. \ref{fig:capacity_Bi}.

In the symmetric path loss cases (i.e., $\gamma_1 = \ldots =
\gamma_K$), two strategies $\mathbf{b}_1$ and $\mathbf{b}_2$ yield
the same performance whenever $\mathbf{b}_1^\downarrow =
\mathbf{b}_2 ^\downarrow$. Thus, the optimal feedback strategy can
be found in the strategy set $\mathcal{B}$ given by
\begin{align}
 \mathcal{B}
    =\Big\{ \mathbf{b}^\downarrow ~\Big\vert~ \mathbf{b} \in (\mathbb{Z}^+\cup
    \{0\})^K,~ \sum_{k=1}^K [\mathbf{b}]_k = K \bar{b} \Big\}.
 \label{eqn:FB_set}
\end{align}
The number of all possible strategies is determined by the total
feedback size as in Table \ref{tab:strategy_number}.

\begin{table}[!b]\centering
\caption{The number of feedback strategies for $4\times 4$ MIMO BC}
\begin{tabular}{c||cccccccc}
 \hline
    Total FB Size   & 8 & 16 & 24 & 32 & 40 & 48 & 56 & 64\\
 \hline
    $\vert\mathcal{B} \vert$   &15 & 64 & 169 & 351 & 632 & 1033 & 1575 & 2280\\
 \hline
\end{tabular}
\label{tab:strategy_number}
\end{table}

For asymmetric path loss cases, without loss of generality we
consider the case that $\gamma_1\ge\ldots\ge\gamma_K$. Because the
larger feedback size yields the higher multiplexing gain, larger
feedback size should be assigned to the user with smaller path loss
(i.e., larger $\gamma_k$). This implicates that the strategy
$\mathbf{b}^\downarrow$ outperforms $\mathbf{b}$, i.e.,
\begin{align}
 \sum_{k=1}^K \mathcal{R}_k(\gamma_kP, [\mathbf{b}^\downarrow]_k)
 \ge \sum_{k=1}^K \mathcal{R}_k(\gamma_kP, [\mathbf{b}]_k).\NN
\end{align}
Therefore, the optimal feedback rate sharing strategy is selected in
the feedback strategy set $\mathcal{B}$ defined in
\eqref{eqn:FB_set}. Because the number of all possible strategies,
i.e., $\vert \mathcal{B} \vert$, is the same for the symmetric and
the asymmetric path loss cases, the computational complexity is also
the same for both cases.

%%%%%%%%%%%%%%%%%%%%%%%%%%%%%%%%%%%%%%%%%%%%%%%%%%%%%%%%%%%%%%%%%%%
% % % % % % % % % % % % % % % % % % % % % % % % % % % % % % % % % %
%%%%%%%%%%%%%%%%%%%%%%%%%%%%%%%%%%%%%%%%%%%%%%%%%%%%%%%%%%%%%%%%%%%
\subsection{Extension to Stream Control}\label{section:stream_control}

Although the equal power allocation with full multiplexing is mainly
considered in our manuscript, our feedback rate sharing strategy can
readily be extended to the stream control where the transmitter
adaptively controls multiplexing gain. For $4\times 4$ MIMO BC, for
example, four ways of equal power allocation according to the number
of streams -- $[P/4, P/4, P/4, P/4]$, $[P/3, P/3, P/3, 0]$, $[P/2,
P/2, 0, 0]$, and $[P, 0, 0, 0]$ -- are possible with the steam
control. Note that single stream transmission corresponds to the
TDMA scheme. Since we consider ZF beamforming at the transmitter,
the beamforming vector for each user is randomly picked orthogonal
to other users' quantized channels. Therefore, it can easily be
shown that Theorem \ref{theorem:strategy_low} and Theorem
\ref{theorem:strategy_high} are still valid even with the stream
control. In Table \ref{tab:strategy}, we have found the optimal
feedback rate sharing strategy for $4\times 4$ MIMO BC according to
the number of streams and SNR when total feedback budget is 24bits
and the path losses are symmetric. We can also find the optimal
feedback rate sharing strategies for asymmetric path losses because
Lemma \ref{lemma:capacity_k} still holds for the stream control and
hence the rate of each served user is affected by its own feedback
size.

%\linespread{1.8}
%%%%%%%%%%%%%%%%%%%%%%%%%%%%%%%%%%%%%%%%%%%%%%%%%%%%%%%%%%%%%%%%%%%
% % % % % % % % % % % % % % % % % % % % % % % % % % % % % % % % % %
%%%%%%%%%%%%%%%%%%%%%%%%%%%%%%%%%%%%%%%%%%%%%%%%%%%%%%%%%%%%%%%%%%%
\section{Numerical Results}
\subsection{Numerical Examples}

In this section, we present numerical results to analyze the effects
of feedback rate sharing strategies. In Fig. \ref{fig:sum_rate}, the
average sum rates of a $2\times 2$ MIMO BC using different feedback
rate sharing strategies.
We consider five feedback rate sharing strategies $(\mathbf{b}_1,
\mathbf{b}_2, \mathbf{b}_3, \mathbf{b}_4, \mathbf{b}_5) = ([0,16],
[2,14], [4,12], [6,10], [8,8]$) such that
$\mathbf{b}_1 \succeq \mathbf{b}_2 \succeq \mathbf{b}_3 \succeq
\mathbf{b}_4 \succeq \mathbf{b}_5$.
In Fig. \ref{fig:sum_rate}, for all $\mathbf{b}_i \succeq
\mathbf{b}_j$ we obtain
$\lim_{P\to 0} \mathcal{R}(P,\mathbf{b}_i) < \lim_{P\to 0}
\mathcal{R}(P,\mathbf{b}_j)$ and $\lim_{P\to \infty}
\mathcal{R}(P,\mathbf{b}_i) > \lim_{P\to \infty}
\mathcal{R}(P,\mathbf{b}_j)$  as stated in Theorem
\ref{theorem:strategy_low} and Theorem \ref{theorem:strategy_high},
respectively.
In the low SNR region, the equal sharing of the sum feedback rate
$\mathbf{b}_5=[8,8]$ achieves the highest average sum rate while
allocating the whole feedback rate to a single user
$\mathbf{b}_1=[0,16]$ achieves the lowest average sum rate as
predicted in Corollary \ref{corollary:strategy_low}.
In the high SNR region, however, allocating the whole feedback rate
to a single user $\mathbf{b}_1=[0,16]$ achieves the highest
achievable rate whereas equal sharing of the feedback rate
$\mathbf{b}_5=[8,8]$ achieves the worst achievable rate as claimed
in Corollary \ref{corollary:strategy_high}.

In a noise limited environment, increasing multiplexing gains
directly results in higher sum rate, and the multiplexing gains are
maximized when the feedback rate is equally shared among users.
Since the remaining interference caused by the quantization error
becomes dominant in the high SNR region, the full multiplexing gain
cannot be achieved and the multiplexing gain rather diminishes as
SNR increases. Therefore, by allocating the whole feedback rate to a
single user, the other users can effectively eliminate the
interference limitation by removing all multiuser interference from
the user being allocated the whole feedback rate. Reducing the
number of interferers is more effective in an interference limited
environment from a sum rate perspective since the multiplexing gain
is already lost.

The sum rate of a $4\times 4$ MIMO BC for various feedback sizes is
shown in Fig.\ref{fig:sum_rate_4user} where the total feedback rate
is restricted to 36 bits.
Four feedback rate sharing strategies are considered --
$(\mathbf{b}_1,\mathbf{b}_2, \mathbf{b}_3, \mathbf{b}_4)$ =
($[0,0,0,36]$, $[0,0,18,18]$, $[0,12,12,12]$, $[9,9,9,9]$) such that
$\mathbf{b}_1 \succeq \mathbf{b}_2 \succeq \mathbf{b}_3 \succeq
\mathbf{b}_4$.
As stated in Theorem \ref{theorem:strategy_low} and Theorem
\ref{theorem:strategy_high}, we can observe that $\lim_{P\to 0}
\mathcal{R}(P,\mathbf{b}_i) < \lim_{P\to 0}
\mathcal{R}(P,\mathbf{b}_j)$ and $\lim_{P\to \infty}
\mathcal{R}(P,\mathbf{b}_i) > \lim_{P\to \infty}
\mathcal{R}(P,\mathbf{b}_j)$ whenever $\mathbf{b}_i \succeq
\mathbf{b}_j$.
Also, we can observe that the equal allocation to the optimal number
of users according to SNR becomes the optimal strategy in the
mid-SNR region as stated in Observation \ref{observation:mid_SNR}.

%%%%%%%%%%%%%%%%%%%%%%%%%%%%%%%%%%%%%%%%%%%%%%%%%%%%%%%%%%%%%%%%%%%
% % % % % % % % % % % % % % % % % % % % % % % % % % % % % % % % % %
%%%%%%%%%%%%%%%%%%%%%%%%%%%%%%%%%%%%%%%%%%%%%%%%%%%%%%%%%%%%%%%%%%%
\subsection{Extension to Other Codebook Models}

Although the overall trends obtained by RVQ are known to agree well
with the results of other codebooks, we consider another codebook
model to verify the observations and conclusions obtained for RVQ
are effective for other codebook models.
Since a rate maximizing codebook is difficult to find, we consider a
spherical cap codebook \cite{MSEA2003, WSS2006,YJG2007}
which is based on an ideal assumption that each quantization cell in
$b$-bit codebook is a spherical cap with the surface area $2^{-b}$.
A spherical cap codebook is an ideal vector quantizer whose
quantization error is stochastically dominated by any other
codebooks \cite{J2006}.
In a $b$-bit spherical cap codebook, the CDF of the quantization
error denoted by $\tilde{Z}$ becomes %\cite{YJG2007}
\begin{align}
    \PR{\tilde{Z}<z}=\bigg\{
    \begin{array}{ll}
        2^{b}z^{M-1},&\quad 0\le z\le 2^{-\frac{b}{M-1}} \NNL
        1,&\quad z\ge 2^{-\frac{b}{M-1}}.
    \end{array}\NN
\end{align}
Fig. \ref{fig:sum_rate_4user} and Fig.
\ref{fig:sum_rate_optimal_codebook} show the average sum rates of a
$4\times 4$ MIMO BC using various feedback sharing strategies when
RVQ and a spherical cap codebook are used, respectively. This result
confirms the optimal strategies obtained from RVQ is still valid for
the spherical cap codebook.

In general, RVQ and spherical cap codebook are regarded as the lower
bound and the upper bound of the practical quantization codebook,
respectively. From the both codebook models, therefore, we can
conjecture the average sum rate in practical $4\times 4$ ZF MIMO BC
for the given configuration.
In Fig. \ref{fig:practical_codebook_case}, the conjectured average
sum rate region for practical quantization codebook (with $\sum
b_k=36$) is shaded with/without adopting our proposed feedback rate
sharing strategy, respectively.
Each shaded region is bounded both on RVQ and the spherical cap
cases plotted in Fig. \ref{fig:sum_rate_4user} and Fig.
\ref{fig:sum_rate_optimal_codebook}, respectively.
Fig. \ref{fig:practical_codebook_case} implicates that our proposed
feedback rate sharing strategy is useful even for practical ZF MIMO
BC systems, especially in the high SNR region.

%%%%%%%%%%%%%%%%%%%%%%%%%%%%%%%%%%%%%%%%%%%%%%%%%%%%%%%%%%%%%%%%%%%
% % % % % % % % % % % % % % % % % % % % % % % % % % % % % % % % % %
%%%%%%%%%%%%%%%%%%%%%%%%%%%%%%%%%%%%%%%%%%%%%%%%%%%%%%%%%%%%%%%%%%%
\subsection{Comparison with TDMA and Regularized ZF}

We also consider the regularized zero-forcing beamforming
\cite{J2006} which enhances the performance of ZF beamforming in the
low SNR region. Also, TDMA is considered and compared with both ZF
beamforming and regularized ZF beamforming.
The average sum rates of a $4\times4$ MIMO BC using ZF beamforming
adopting our proposed feedback rate sharing strategy are compared
with TDMA in Fig. \ref{fig:sum_rate_TDMA}, when $\sum b_k =60$.
In TDMA, all available feedback bits are allocated to the single
served user ($\mathbf{b}=[60]$).
In Fig. \ref{fig:sum_rate_TDMA}, we can observe that ZF beamforming
is inferior to a TDMA system in both low and high SNR regions
although it outperforms a TDMA system in the mid SNR region. In
these regions, it is desirable to adopt the mode switching
\cite{ZHKA2009} between ZF and TDMA for sum rate maximization.

In the regularized ZF beamforming, the normalized column vectors of
$\hat{\mathbf{H}} ^\dagger \left(\hat{\mathbf{H}}
\hat{\mathbf{H}}^\dagger + \frac{M}{P} \mathbf{I}_M \right)^{-1}$
are used for the beamforming vectors where $\mathbf{I}_M$ is an
$M\times M$ identity matrix.
Although the optimal feedback rate sharing strategy using the
regularized ZF beamforming is hard to analyze, the feedback rate
sharing strategy will be the same with that of ZF beamforming case
in the high SNR region. This is because the regularized ZF
beamforming vectors correspond to ZF beamforming vectors in the high
SNR region.
In Fig. \ref{fig:RVQ_TDMA_MMSE_60bits}, the average sum rates of a
$4\times4$ MIMO BC using regularized ZF beamforming are plotted
while other parameters are same in Fig. \ref{fig:sum_rate_TDMA}.
As shown in Fig. \ref{fig:RVQ_TDMA_MMSE_60bits}, the regularized ZF
beamforming improves ZF beamforming especially in the low SNR region
and hence outperforms TDMA in wider SNR region.

Since TDMA always achieves a multiplexing gain of one even with
blind transmission, TDMA system outperforms MIMO BC with limited
feedback in the high SNR region.
This is because the achievable rate of MIMO BC with finite limited
feedback is saturated in the high SNR region due to mutual
interference.
The inferior performance in the high SNR region is a fundamental
limit of MIMO BC with limited feedback.
However, it should be noted that ZF beamforming can be enhanced by
the regularized ZF beamforming and our feedback rate sharing
strategy enables ZF beamforming or regularized ZF beamforming to
outperform TDMA in wider SNR region.
%
%
%In Fig. \ref{fig:RVQ_TDMA_MMSE}, sum rates of a $4\times 4$ MIMO BC
%using regularized ZF beamforming when $\sum b_k=60$, $\sum b_k=48$
%and $\sum b_k=36$, respectively. In these figures, we can observe
%that the feedback rate sharing strategy yields more significant
%performance improvement when the feedback rate is large.
%
Note that our main contributions are to find the feedback rate
sharing strategy and to show the feedback rate sharing strategy
(e.g., $\sum b_k=60$) enhances the system performance compare to
equal feedback rate sharing (e.g. $\mathbf{b}=[15,15,15,15]$).
In Fig. \ref{fig:RVQ_TDMA_MMSE_60bits}, the regularized ZF
beamforming outperforms TDMA from -15dB to about 45dB when the
optimal feedback rate sharing strategy is employed, whereas equally
sharing makes the regularized ZF beamforming outperform TDMA until
about 34dB.

%%%%%%%%%%%%%%%%%%%%%%%%%%%%%%%%%%%%%%%%%%%%%%%%%%%%%%%%%%%%%%%%%%%
% % % % % % % % % % % % % % % % % % % % % % % % % % % % % % % % % %
%%%%%%%%%%%%%%%%%%%%%%%%%%%%%%%%%%%%%%%%%%%%%%%%%%%%%%%%%%%%%%%%%%%
\section{Conclusion}
In this paper, we have analyzed the average sum rate of ZF MIMO BC
with limited feedback when the users share the feedback rates.
The impact of asymmetric feedback sizes among the users has been
rigorously analyzed by adopting RVQ at each user.
Our mathematical analysis has shown that the optimal feedback rate
sharing strategy in the high SNR region is to allocate the whole
feedback rate to a single user.
On the other hand, the optimal feedback rate sharing strategy in the
low SNR region is the equal sharing of the feedback rate among
users.
We have proposed a simple numerical method for finding the optimal
feedback rate sharing strategy in the practical SNR region and shown
that equal sharing of the feedback rate among the optimal number of
concurrent users is optimal. It has also been shown that the
proposed numerical method can be applicable to finding the optimal
feedback rate sharing strategy when path losses of the users are
different.
In the simulation part, we have shown our proposed feedback capacity
sharing strategy is still valid for other system configurations such
as regularized zeroforcing transmission and spherical-cap codebook.

%%%%%%%%%%%%%%%%%%%%%%%%%%%%%%%%%%%%%%%%%%%%%%%%%%%%%%%%%%%%%%%%%%%%%
%\appendices
%\def\thesection{\Alph{section}}%
%\def\thesectiondis{\Alph{section}}%
%%%%%%%%%%%%%%%%%%%%%%%%%%%%%%%%%%%%%%%%%%%%%%%%%%%%%%%%%%%%%%%%%%%%%

%%%%%%%%%%%%%%%%%%%%%%%%%%%%%%%%%%%%%%%%%%%%%%%%%%%%%%%%%%%%%%%%%%%%%
\section*{Appendix A. Proof of Lemma \ref{lemma:RV_properties}}
\setcounter{equation}{0}
\renewcommand{\theequation}{A.\arabic{equation}}
%\label{appendix:RV_properties}
%%%%%%%%%%%%%%%%%%%%%%%%%%%%%%%%%%%%%%%%%%%%%%%%%%%%%%%%%%%%%%%%%%%%%
Since the channel vectors are i.i.d, it is obvious that $Q_k \sim
Q_1$ for all $k$.
Because $\mathbf{h}_k$ is isotropic in $\mathbb{C}^{M}$, the
quantization of $\mathbf{h}_k$ is also isotropic in
$\mathbb{C}^{M}$.
Thus, $\{\hat{\mathbf{h}}_k\} _{k=1}^K$ become independent and
isotropically distributed random vectors in $\mathbb{C}^{M}$.
Because $\mathbf{v}_k$ is uniquely obtained from
$\{\hat{\mathbf{h}}_i\} _{i\ne k}$, the beamforming vectors,
$\{\mathbf{v}_k\} _{k\ne 1}^K$, are also isotropic in
$\mathbb{C}^M$.
Since $\mathbf{v}_k$ is independent of $\hat{\mathbf{h}}_k$, $X_k
(=\vert \tilde{\mathbf{h}}_k \mathbf{v}_k \vert^2)$ becomes the
squared inner product between two independent random vectors
isotropic in $\mathbb{C}^M$. Hence, $X_k$ is identical for all $k$,
i.e., $X_k \sim X_1$.
For $W_k (= \sum_{i\ne k}| \mathbf{e}_k ^\dagger \mathbf{v}_i|^2)$,
both $\mathbf{e}_k$ and $\{\mathbf{v}_i\}_{i\ne k}$ are picked
independently in the null space of $\hat{\mathbf{h}}_k$, and they
are also isotropic in the $M-1$ dimensional subspace. Thus, $W_k$
becomes the sum of $K-1$ the squared inner products between two
independent and isotropic random vectors in the $M-1$ dimensional
subspace in $\mathbb{C}^M$ so that $W_k \sim W_1$, $\forall k$.
From above reasons, we can conclude that $Q_k$, $X_k$, and $W_k$ are
identical for all $k$, respectively, invariant with the feedback
sizes $b_1, \ldots, b_K$.

We can prove the second property that $\{Q_k, X_k, W_k\}_{k=1}^K$ is
independent of all $\{Z_k\}_{k=1}^K$ because $Z_k$ is only dependent
on $b_k$ as shown in \eqref{eqn:Z_k}.

Because $\{Q_i, X_i, W_i\}$ is interchangebly obtained from $\{Q_k,
X_k, W_k\}$ by swapping the index of $\mathbf{h}_i$ and
$\mathbf{h}_k$ whose distribution are the same, i.e., $Q_i \sim
Q_k$, $X_i \sim X_k$, and $W_i \sim W_k$, we can obtain the third
property such that
\begin{align}
 f_{Q_k, X_k, W_k}(q, x, w) = f_{Q_1, X_1, W_1}(q, x, w),
    \quad  k = 1, \ldots, K. \NN
\end{align}

When all users use the equal feedback size, (i.e., $Z_k \sim
\bar{Z}$, $\forall k$), the average achievable rate of each user is
the same such that $\mathbb{E}\left[ \log_2\left( 1 + \frac{
\frac{P}{M} Q_k X_k}{1+ \frac{P}{M} Q_k W_k \bar{Z}} \right) \right]
=\mathbb{E}\left[ \log_2\left( 1+\frac{\frac{P}{M} Q_1X_1}{1+
\frac{P}{M} Q_1 W_1 \bar{Z}}\right) \right]$ for all $k$.
%
%\begin{align}
% \mathbb{E}\left[ \log_2\left(
%    1+\frac{\frac{P}{M} Q_kX_k}{1+ \frac{P}{M} Q_k W_k \bar{Z}}
%    \right) \right]
% =\mathbb{E}\left[ \log_2\left(
%    1+\frac{\frac{P}{M} Q_1X_1}{1+ \frac{P}{M} Q_1 W_1 \bar{Z}}
%    \right) \right], \quad \forall k.
%\end{align}
%
This can be explained from the fact that $f_{Q_k, X_k, W_k,
\bar{Z}}(q, x, w, z) \stackrel{(a)}{=} f_{Q_k, X_k, W_k}(q, x, w)
f_{\bar{Z}}(z) \stackrel{(b)}{=} f_{Q_1, X_1, W_1}(q, x, w)
f_{\bar{Z}}(z) \stackrel{(a)}{=}f_{Q_1, X_1, W_1, \bar{Z}}(q, x, w,
z)$ where $(a)$ and $(b)$ are from the second property and the third
property, respectively.

%
%\begin{align}
% &f_{Q_k, X_k, W_k, \bar{Z}}(q, x, w, z)
%    \stackrel{(a)}{=}f_{Q_k, X_k, W_k}(q, x, w) f_{\bar{Z}}(z) \NNL
%    &\qquad\qquad\stackrel{(b)}{=} f_{Q_1, X_1, W_1}(q, x, w) f_{\bar{Z}}(z)
%    \stackrel{(a)}{=}f_{Q_1, X_1, W_1, \bar{Z}}(q, x, w, z) \quad\forall k,
%\end{align}
%where the equalities $(a)$ and $(b)$ are from the second property
%and the third property, respectively.

%%%%%%%%%%%%%%%%%%%%%%%%%%%%%%%%%%%%%%%%%%%%%%%%%%%%%%%%%%%%%%%%%%%%%%
\section*{Appendix B. Proof of Theorem \ref{theorem:strategy_low}}
\setcounter{equation}{0}
\renewcommand{\theequation}{B.\arabic{equation}}
\label{appendix:strategy_low}
%%%%%%%%%%%%%%%%%%%%%%%%%%%%%%%%%%%%%%%%%%%%%%%%%%%%%%%%%%%%%%%%%%%%%%

To prove Theorem \ref{theorem:strategy_low}, we firstly show the
average quantization error $\mathbb{E}[Z_k]$ is a discretely convex
function of $b_k$. Then, we use the majorization theory.
We start from following Lemma.

%%%%%%%%%%%%%%%%%%%%%%%%%%%%%%%%%%%%%%%%%%%%%%%%%%%%%%%%%%%%%%%%%%%
% % % % % % % % % % % % % % % % % % % % % % % % % % % % % % % % % %
%%%%%%%%%%%%%%%%%%%%%%%%%%%%%%%%%%%%%%%%%%%%%%%%%%%%%%%%%%%%%%%%%%%
\begin{lemma}\label{lemma:EZ_convexity}
The average quantization error $\mathbb{E}[Z_k]$ is a discretely
convex function of $b_k$.
\end{lemma}
\begin{proof}
It was shown in \cite{J2006, AL2007} that $\mathbb{E}[Z_k \vert b_k
= b] =2^{b} \cdot \beta \left(2^{b}, \frac{M}{M-1}\right)$,
%
%\begin{align}
%    \mathbb{E}[Z_k \vert b_k = b]
%    =2^{b} \cdot \beta \left(2^{b}, \frac{M}{M-1}\right),
%\end{align}
%
where $\beta(x,y)$ is the beta function given by
$\beta(x,y)=\frac{\Gamma(x)\Gamma(y)}{\Gamma(x+y)}$.
Using this, we obtain
\begin{align}
 &\mathbb{E}[Z_k \vert b_k = b+1]\NNL
 &= 2^{b+1} \cdot \beta\left(2^{b+1}, \frac{M}{M-1}\right) \NNL
 &= \frac{2\cdot\Gamma\left(2^{b+1}\right)\Gamma \left(2^b + \frac{M}{M-1}\right)}
    {\Gamma\left(2^b\right)\Gamma \left(2^{b+1} + \frac{M}{M-1}\right)}
 \times \frac{2^b \cdot\Gamma\left(2^b\right)\Gamma\left(\frac{M}{M-1}\right)}
    {\Gamma \left(2^b + \frac{M}{M-1}\right)}\NNL
 &\stackrel{(a)}{=} \frac{2 \cdot\prod_{i=2^b} ^{2^{b+1}-1} i} {\prod_{i=2^b}^{2^{b+1}-1}
    \big(i+\frac{M}{M-1}\big)} \times \mathbb{E}[Z_k \vert b_k = b],\NN
\end{align}
where the equality $(a)$ is from $\Gamma(x+1)=x\Gamma(x)$. Thus, we
can rewrite $\mathbb{E}[Z_k \vert b_k = b+1] = \eta_b \cdot
\mathbb{E}[Z_k \vert b_k = b]$ where $\eta_b\triangleq 2 \cdot \prod
_{i=2^b} ^{2^{b+1}-1} \frac{i}{\big(i+\frac{M}{M-1}\big)}$.

When we define a forward difference function $\Delta (b) \triangleq
\mathbb{E}[Z_k\vert b_k=b+1] - \mathbb{E}[Z_k\vert b_k=b]$, we can
find that the forward difference function is an increasing function
of $b$, i.e., $\Delta (b+1) > \Delta (b)$, such that
\begin{align}
&\Delta (b+1) - \Delta (b)\NNL
 & = \mathbb{E}[Z_k \vert b_k= b+2 ]
    - 2\cdot\mathbb{E}[Z_k \vert b_k=b+1]
    + \mathbb{E}[Z_k \vert b_k=b] \NNL
 & = (\eta_{b+1}\eta_b - 2\eta_b + 1)\cdot
    \mathbb{E}[Z_k\vert b_k=b] %\NNL
 \stackrel{(a)}{>} 0 \NN
\end{align}
where $(a)$ is from the fact that $\eta_{b+1}\eta_b - 2\eta_b = 4
\cdot \left(\prod_{i=2^b} ^{2^{b+2}-1}
\frac{i}{\big(i+\frac{M}{M-1}\big)} - \prod_{i=2^b} ^{2^{b+1}-1}
\frac{i} {\big(i + \frac{M} {M-1}\big)}\right)$ is ranged in $[-1,
0]$ and minimized and maximized when $M=2$ and $M=\infty$,
respectively.
Since a discretely convex function has an increasing
(non-decreasing) forward difference function \cite{Y2002},
$\mathbb{E}[Z_k]$ is a discretely convex function of $b_k$.
\end{proof}
%%%%%%%%%%%%%%%%%%%%%%%%%%%%%%%%%%%%%%%%%%%%%%%%%%%%%%%%%%%%%%%%%%%
% % % % % % % % % % % % % % % % % % % % % % % % % % % % % % % % % %
%%%%%%%%%%%%%%%%%%%%%%%%%%%%%%%%%%%%%%%%%%%%%%%%%%%%%%%%%%%%%%%%%%%
%
It is widely known in majorization theory that for a convex function
$h:\mathbb{R}\to\mathbb{R}$ and two vectors $\mathbf{a}_1,
\mathbf{a}_2 \in \mathbb{R}^n$, %it satisfies
\begin{align}
 \sum_{i=1}^n h([\mathbf{a}_1]_i) \le \sum_{i=1}^n
 h([\mathbf{a}_2]_i),
\end{align}
whenever $\mathbf{a}_1 \preceq \mathbf{a}_2$.
In the low SNR region, the sum average rate with feedback rate
sharing strategy is only related with $\sum_{k=1}^K \mathbb{E}[Z_k]$
as stated in \eqref{eqn:optimization_problem_L}.
From Lemma \ref{lemma:EZ_convexity}, we know the average
quantization error is a convex function of $b_k$.
With the feedback rate sharing strategies $\mathbf{b}_1 \preceq
\mathbf{b}_2$, therefore, we can conclude that
\begin{align}
    \sum_{k=1}^K \mathbb{E}\{Z_k\vert b_k=[\mathbf{b}_1]_k\} \le
    \sum_{k=1}^K \mathbb{E}\{Z_k\vert b_k=[\mathbf{b}_2]_k\},
\end{align}
and equivalently, $\lim_{P\to 0} \mathcal{R}(P,\mathbf{b}_1) >
\lim_{P\to 0} \mathcal{R} (P,\mathbf{b}_2)$.

%%%%%%%%%%%%%%%%%%%%%%%%%%%%%%%%%%%%%%%%%%%%%%%%%%%%%%%%%%%%%%%%%%%%%%
\section*{Appendix C. Proof of Theorem \ref{theorem:strategy_high}}
\setcounter{equation}{0}
\renewcommand{\theequation}{C.\arabic{equation}}
\label{appendix:strategy_high}
%%%%%%%%%%%%%%%%%%%%%%%%%%%%%%%%%%%%%%%%%%%%%%%%%%%%%%%%%%%%%%%%%%%%%%
%Proof of Theorem \ref{theorem:strategy_high} is similar with Proof
%of Theorem \ref{theorem:strategy_low}.
%
We firstly show that $\mathbb{E}\left[\log_2 Z_k\right]$ is a
discretely concave function of $b_k$ in following lemma.

%%%%%%%%%%%%%%%%%%%%%%%%%%%%%%%%%%%%%%%%%%%%%%%%%%%%%%%%%%%%%%%%%%%
% % % % % % % % % % % % % % % % % % % % % % % % % % % % % % % % % %
%%%%%%%%%%%%%%%%%%%%%%%%%%%%%%%%%%%%%%%%%%%%%%%%%%%%%%%%%%%%%%%%%%%
\begin{lemma}\label{lemma:ElogZ_concavity}
The average quantization error $\mathbb{E}[\log_2 Z_k]$ is a
discretely concave function of $b_k$.
\end{lemma}
%%%%%%%%%%%%%%%%%%%%%%%%%%%%%%%%%%%%%%%%%%%%%%%%%%%%%%%%%%%%%%%%%%%
% % % % % % % % % % % % % % % % % % % % % % % % % % % % % % % % % %
%%%%%%%%%%%%%%%%%%%%%%%%%%%%%%%%%%%%%%%%%%%%%%%%%%%%%%%%%%%%%%%%%%%
\begin{proof}
In \cite{J2006}, it was shown that $\mathbb{E}\left[\log_2 Z_k\vert
b_k = b\right] = \frac{-\log_2e} {M-1} \sum_{i=1} ^{2^{b}}
\frac{1}{i}$.
%
%\begin{align}
%    \mathbb{E}\left[\log_2 Z_k\vert b_k = b\right]
%    = \frac{-\log_2e}{M-1} \sum_{i=1}^{2^{b}}\frac{1}{i}.
%\end{align}
%
In this case, the forward difference function $\Delta(b) \triangleq
\mathbb{E}\left[\log_2 Z_k \vert b_k = b+1\right] -
\mathbb{E}\left[\log_2 Z_k \vert b_k = b \right]$ becomes
\begin{align}
    \Delta (b)
    &= \frac{-\log_2e}{M-1} \sum_{i=2^b+1}^{2^{(b+1)}} \frac{1}{i},
\end{align}
and is a monotonically decreasing function of $b$, i.e., $\Delta (b)
> \Delta (b+1)$.
Since a discretely concave function has a decreasing(non-increasing)
forward difference function \cite{Y2002}, $\mathbb{E}[\log_2 Z_k]$
is a discretely concave function of $b_k$.
\end{proof}
%%%%%%%%%%%%%%%%%%%%%%%%%%%%%%%%%%%%%%%%%%%%%%%%%%%%%%%%%%%%%%%%%%%
% % % % % % % % % % % % % % % % % % % % % % % % % % % % % % % % % %
%%%%%%%%%%%%%%%%%%%%%%%%%%%%%%%%%%%%%%%%%%%%%%%%%%%%%%%%%%%%%%%%%%%
%
In majorization theory, for a concave function $g:\mathbb{R}
\to\mathbb{R}$, it satisfies that
\begin{align}
 \sum_{i=1}^n g([\mathbf{a}_1]_i) \ge \sum_{i=1}^n g([\mathbf{a}_2]_i)
\end{align}
whenever two vectors $\mathbf{a}_1, \mathbf{a}_2 \in \mathbb{R}^n$
satisfies $\mathbf{a}_1 \preceq \mathbf{a}_2$.
In the high SNR region, the average sum rate with feedback rate
sharing strategy is related with $\sum_{k=1}^K \mathbb{E}[\log_2
Z_k]$ as stated in \eqref{eqn:optimization_problem_H}.
As stated in Lemma \ref{lemma:ElogZ_concavity}, $\mathbb{E}[\log_2
Z_k]$ is the concave function of $b_k$.
Thus, under the feedback rate sharing strategies $\mathbf{b}_1
\preceq \mathbf{b}_2$, we can conclude that
\begin{align}
 \sum_{k=1}^K \mathbb{E}\{\log_2 Z_k\vert b_k=[\mathbf{b}_1]_k\} \ge
 \sum_{k=1}^K \mathbb{E}\{\log_2 Z_k\vert b_k=[\mathbf{b}_2]_k\},\NN
\end{align}
equivalently, $\lim_{P\to \infty} \mathcal{R}^-(P,\mathbf{b}_1)
> \lim_{P\to \infty} \mathcal{R}^- (P,\mathbf{b}_2)$.
As stated in Section \ref{section:high_SNR_region}, in the high SNR
region, the achievable rate at each user is dominated by the rate
decreasing term. Thus, we conclude that the feedback rate sharing
strategy $\lim_{P\to \infty} \mathcal{R}(P,\mathbf{b}_1) <
\lim_{P\to \infty} \mathcal{R} (P,\mathbf{b}_2)$ for feedback rate
sharing strategies $\mathbf{b}_1 \preceq \mathbf{b}_2$.

%%%%%%%%%%%%%%%%%%%%%%%%%%%%%%%%%%%%%%%%%%%%%%%%%%%%%%%%%%%%%%%%%%%
% % % % % % % % % % % % % % % % % % % % % % % % % % % % % % % % % %
%%%%%%%%%%%%%%%%%%%%%%%%%%%%%%%%%%%%%%%%%%%%%%%%%%%%%%%%%%%%%%%%%%%
\vspace{.3in}

\linespread{1.6}

\newpage

\begin{figure}[!t]
    \centerline{\psfig{figure=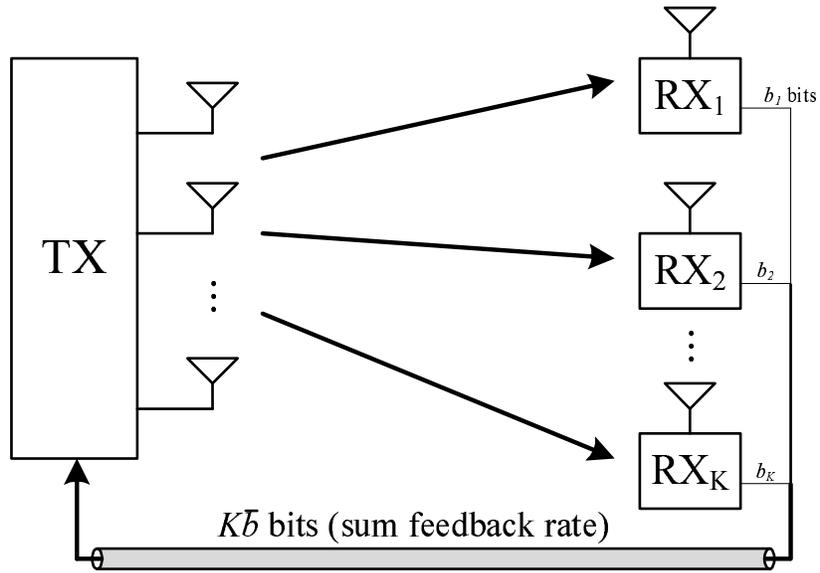,width={0.6\columnwidth}} }
    \caption{A system model. The sum feedback rate is shared by all users.}
    \label{fig1}
\end{figure}

\begin{figure}[!t]
    \centerline{\psfig{figure=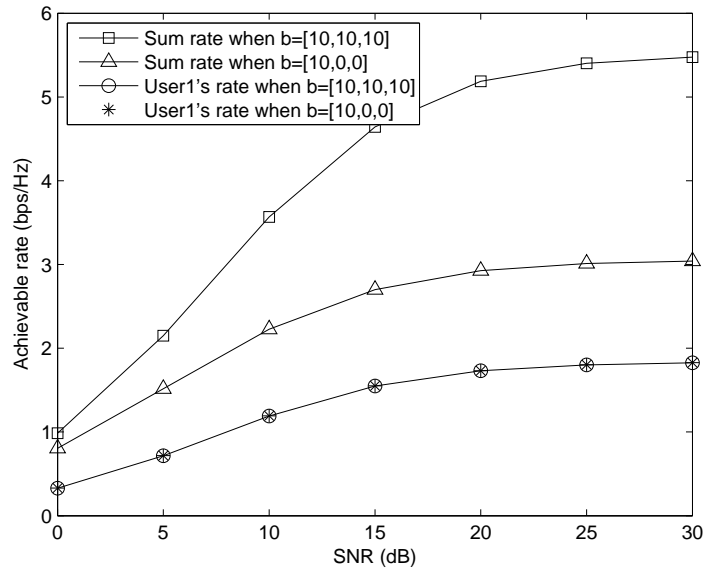,width={0.6\columnwidth}} }
    \caption{The sum rate and the achievable rate at the user 1 in
    $3\times 3$ MIMO BC. The achievable rate of
    user 1 is not affected by the other users' feedback sizes, while the
    sum rate is increased as the feedback sizes of other users increase.}
    \label{fig:capacity_user1}
\end{figure}

%\begin{algorithm}[!t]
%\caption{Procedure to find Feedback Rate Sharing Strategy}
%    \begin{algorithmic}[1]
%    \State{\textbf{Initialization:} randomly choose $\mathbf{b}\in\mathcal{B}$}
%    \ForAll{$\mathbf{b}'\in \mathcal{B}$}
%    \If {$\sum\mathcal{R}_k(\gamma_k P,[\mathbf{b}]_k) > \sum\mathcal{R}_k(\gamma_k P,[\mathbf{b}']_k)$}
%    \State{$\mathbf{b}=\mathbf{b}'$;}
%    \EndIf
%    \EndFor
%    \State{\textbf{Output:} the optimal feedback strategy $\mathbf{b}$}
%    \end{algorithmic}
%    \label{alg:FRS_procedure}
%\end{algorithm}

\begin{figure}[!t]
    \centerline{\psfig{figure=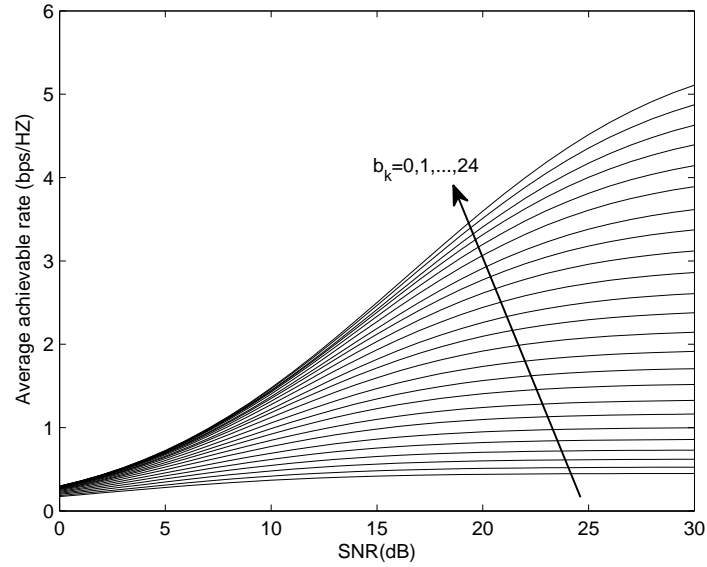,width={0.6\columnwidth}} }
    \caption{Achievable rate of a single user using $b_k$ feedback bits in a $4\times 4$ MIMO BC.}
    \label{fig:capacity_Bi}
\end{figure}

\begin{comment}
\begin{table}[h]\centering
\caption{The optimal feedback rate sharing strategy for a $4\times
4$ MIMO BC}
\begin{tabular}{c|c||c|c||c|c}
 \hline \hline
   \multicolumn{2}{c||}{2 streams} &
   \multicolumn{2}{c||}{3 streams} &
   \multicolumn{2}{c}{4 streams}  \\
 \hline
   SNR{\scriptsize(dB)}& $\mathbf{b}^\downarrow$ & SNR&
   $\mathbf{b}^\downarrow$ &  SNR& $\mathbf{b}^\downarrow$\\
 \hline \hline
     0$\sim$27 & [12,12] & 0$\sim$12 & [8,8,8] & 0$\sim$7 &[6,6,6,6]\\
    28$\sim$~  & [24,~0] & 13$\sim$23 & [12,12,0] & 8$\sim$11 &[8,8,8,0]\\
        ~      &    ~    & 24$\sim$~  & [24,0,0] & 12$\sim$20 &[12,12,0,0]\\
        ~      &    ~    &      ~     &  ~          &   21$\sim$ &[24,0,0,0]\\
 \hline \hline
\end{tabular}
\label{tab:strategy}
\end{table}\end{comment}

\begin{comment}
\begin{table}[b]\centering
    \caption{The optimal feedback rate sharing strategy for a
    $4\times 4$ MIMO BC when all users suffering difference path
    losses $(\gamma_1,\gamma_2, \gamma_3,\gamma_4)=(1.5, 1.25, 1,
    0.75)$}
    \begin{tabular}{c|c||c|c}
    \hline
    SNR  & $[b_1,b_2,b_3,b_4]$  & SNR  & $[b_1,b_2,b_3,b_4]$\\\hline
    $~0\sim~1$ dB & $[8,8,8,0]$ & $8\sim17$ dB& $[13,11,0,0]$ \\\hline
    $~2\sim~6$ dB& $[10,8,6,0]$  & $18$ dB & $[16,8,0,0]$ \\\hline
    $7$ dB& $[11,8,5,0]$        & $19$ dB $\sim$& $[24,0,0,0]$ \\\hline
    \end{tabular}\label{tab:strategy_pathlosses}
\end{table}
\end{comment}

\begin{figure}[!t]
    \centerline{\psfig{figure=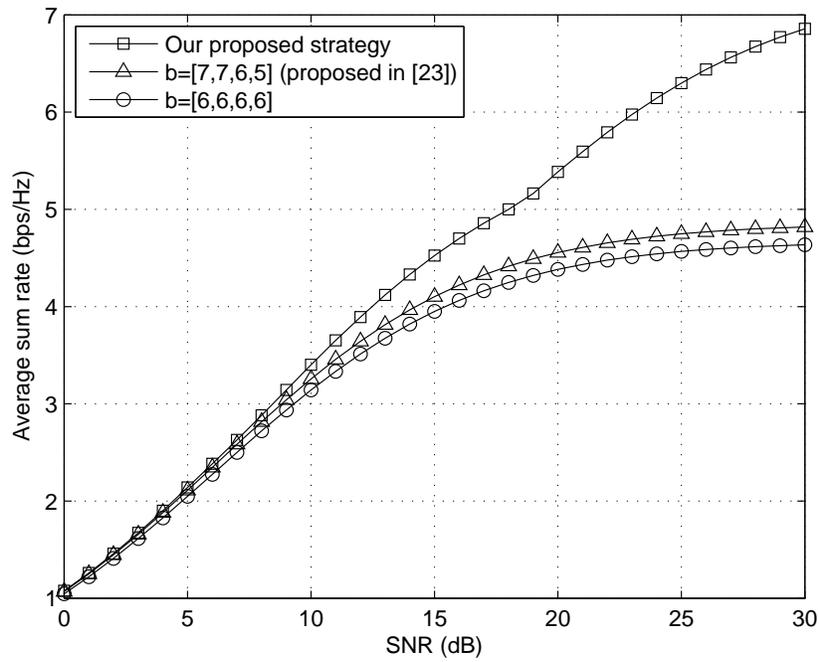,width={0.7\columnwidth}} }
    \caption{Sum rates of a $4\times 4$ MIMO BC using various
    feedback rate sharing strategies ($\sum b_k=24$). Different
    path losses among the users are considered $(\gamma_1,\gamma_2,
    \gamma_3,\gamma_4)=(1.5, 1.25, 1, 0.75)$.}
    \label{fig:capacity_pathlosses}
\end{figure}

\begin{comment}
\begin{table}[!h]\centering
\caption{The number of feedback strategies for $4\times 4$ MIMO BC}
\begin{tabular}{c||cccccccc}
 \hline
    Total FB Size   & 8 & 16 & 24 & 32 & 40 & 48 & 56 & 64\\
 \hline
    $\Vert\mathcal{B} \Vert$   &15 & 64 & 169 & 351 & 632 & 1033 & 1575 & 2280\\
 \hline
\end{tabular}
\label{tab:strategy_number}
\end{table}
\end{comment}

\begin{figure}[!t]
    \centerline{\psfig{figure=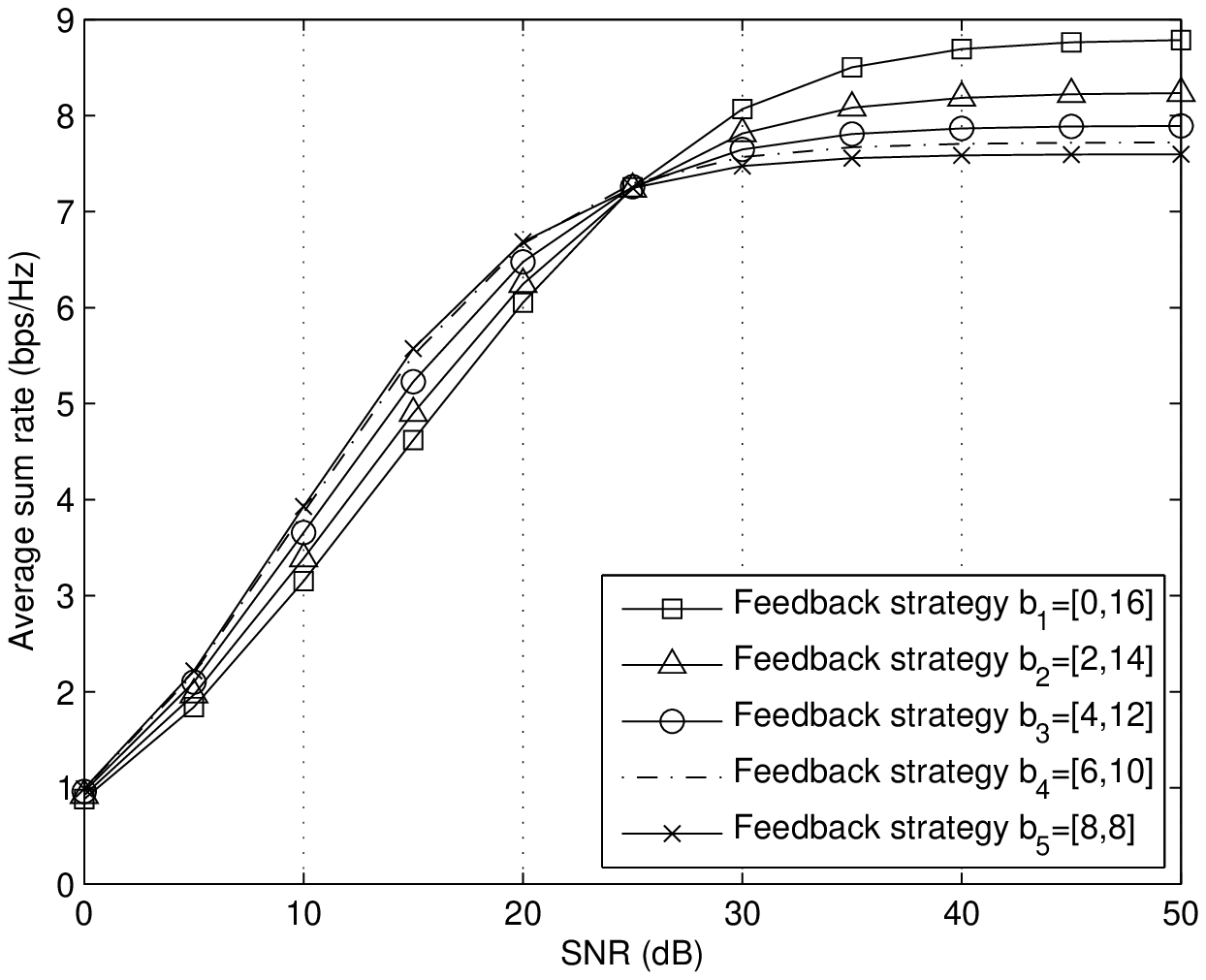,width={0.7\columnwidth}} }
    \caption{Sum rates of a $2\times 2$ MIMO BC using various
    feedback rate sharing strategies ($\sum b_k=16$).}
    \label{fig:sum_rate}
\end{figure}

\begin{figure}[!t]
\centering
\subfigure[Random vector codebook.]
    {\includegraphics[width=.45\columnwidth]{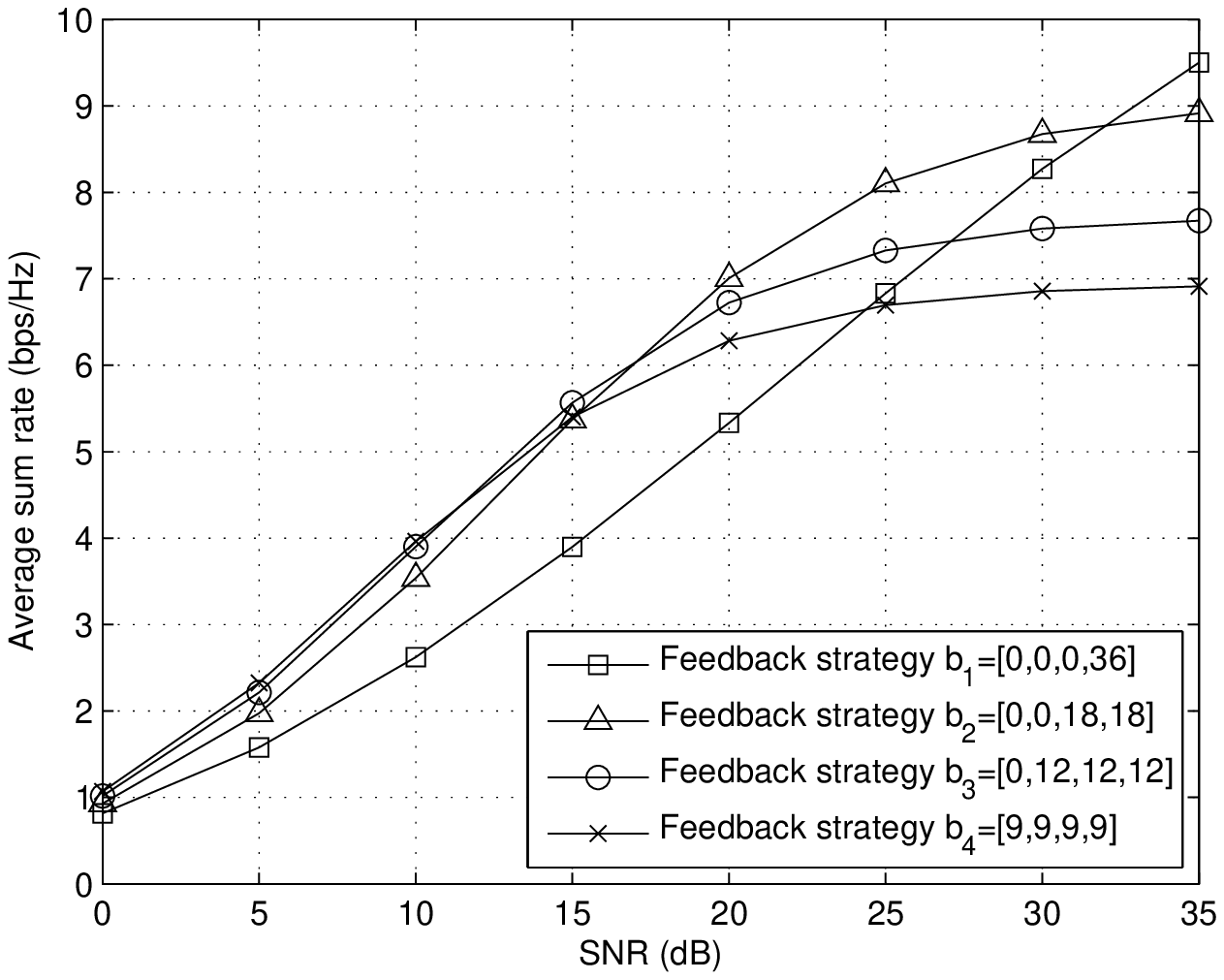}
    \label{fig:sum_rate_4user}}\\
\subfigure[Spherical cap codebook.]
    {\includegraphics[width=.45\columnwidth]{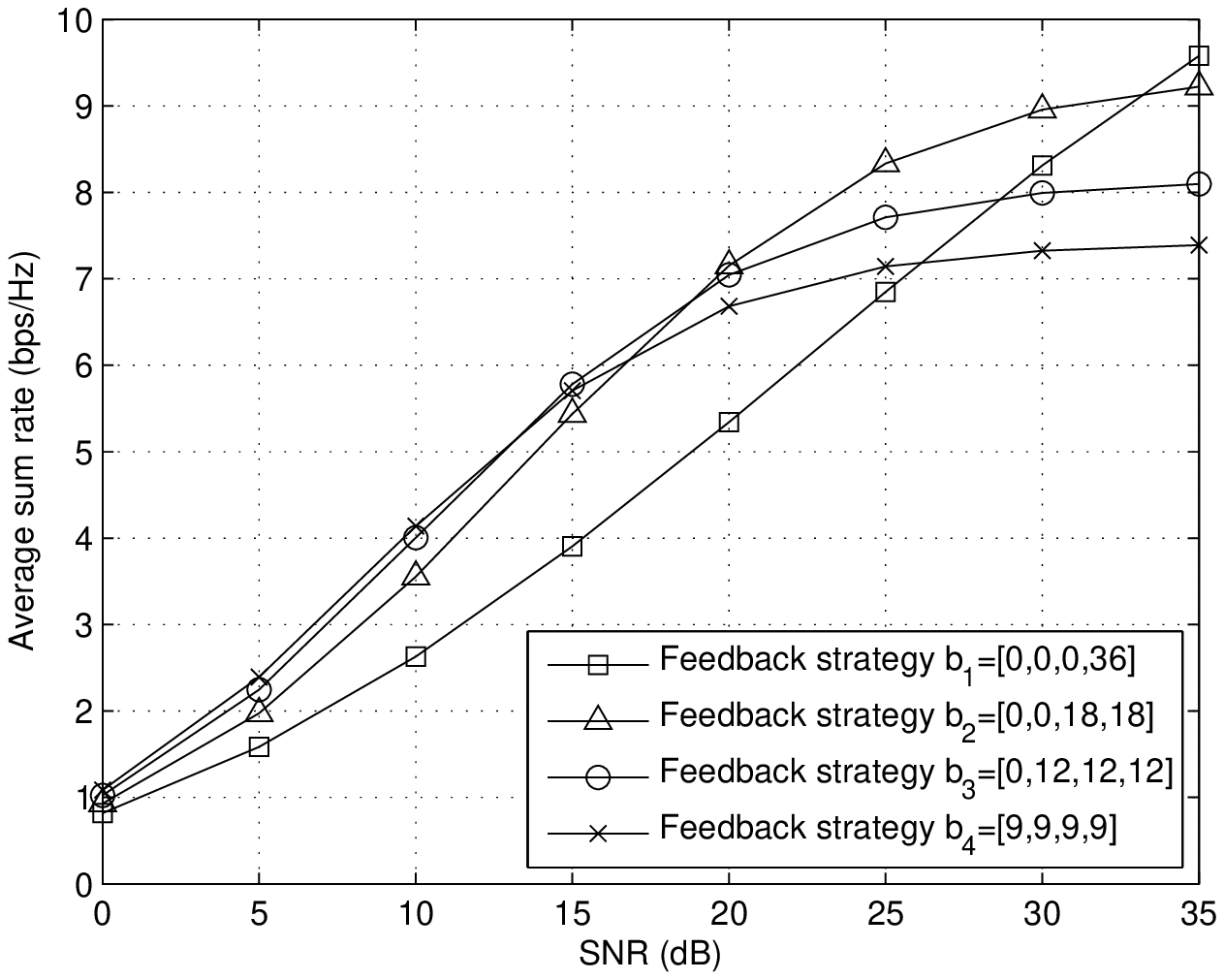}
    \label{fig:sum_rate_optimal_codebook}}\\
\subfigure[The conjectured average sum rate region for practical
quantization codebook.]
    {\includegraphics[width=.45\columnwidth]{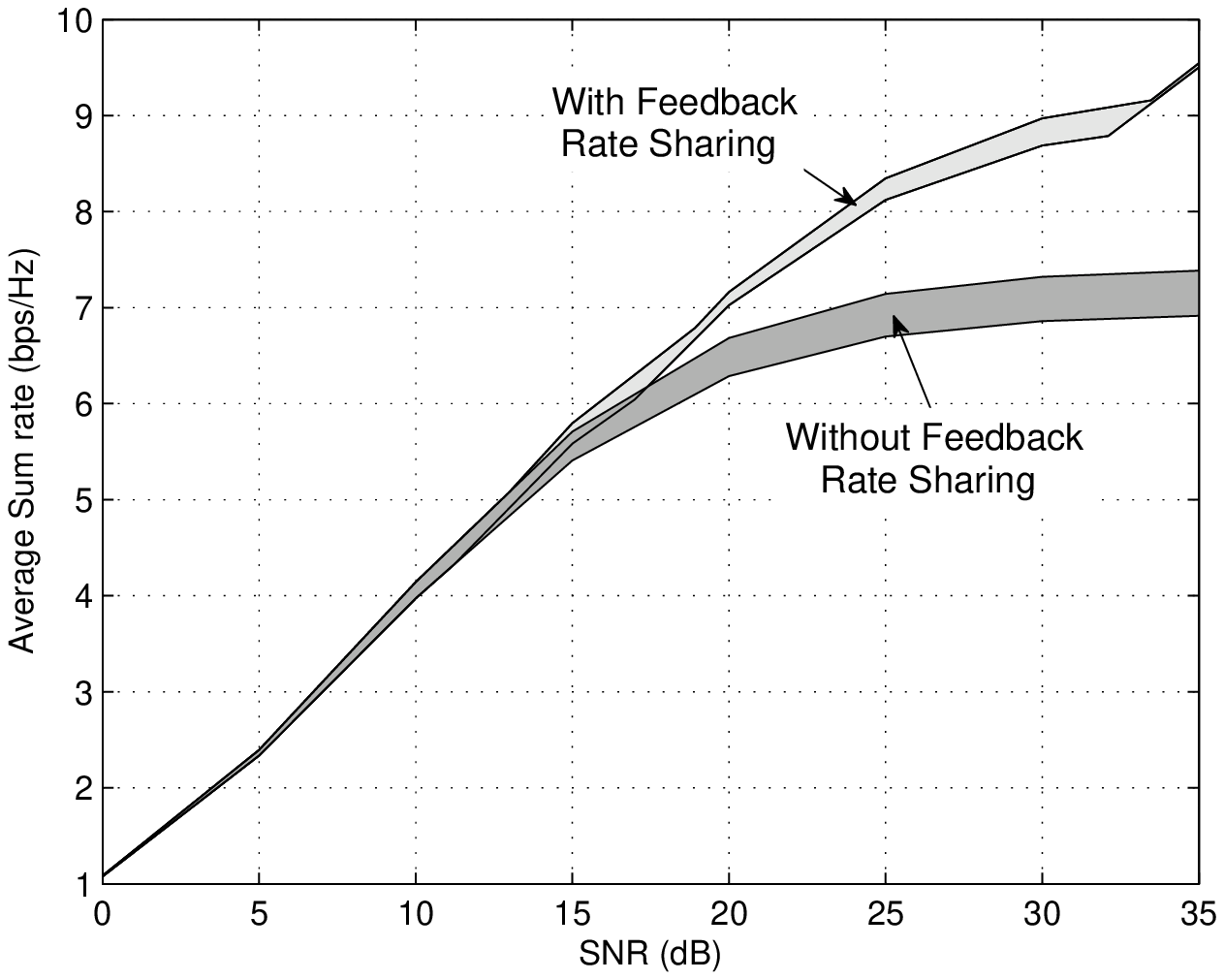}
    \label{fig:practical_codebook_case}}
\caption{Sum rates of a $4\times 4$ MIMO BC using various feedback
rate sharing strategies ($\sum b_k=36$).}
\end{figure}

%\begin{figure}[!t]
%\centerline{\psfig{figure=RVQ_4user_60bit_TDMA_new.eps,width={0.7\columnwidth}}
%} \caption{Sum rates of a $4\times 4$ MIMO BC using various feedback
%rate sharing strategies ($\sum b_k=60$).} \label{fig:sum_rate_TDMA}
%\end{figure}

\begin{figure}[!t]
\centering
\subfigure[ZF beamforming vs. TDMA]
    {\includegraphics[width=.51\columnwidth]{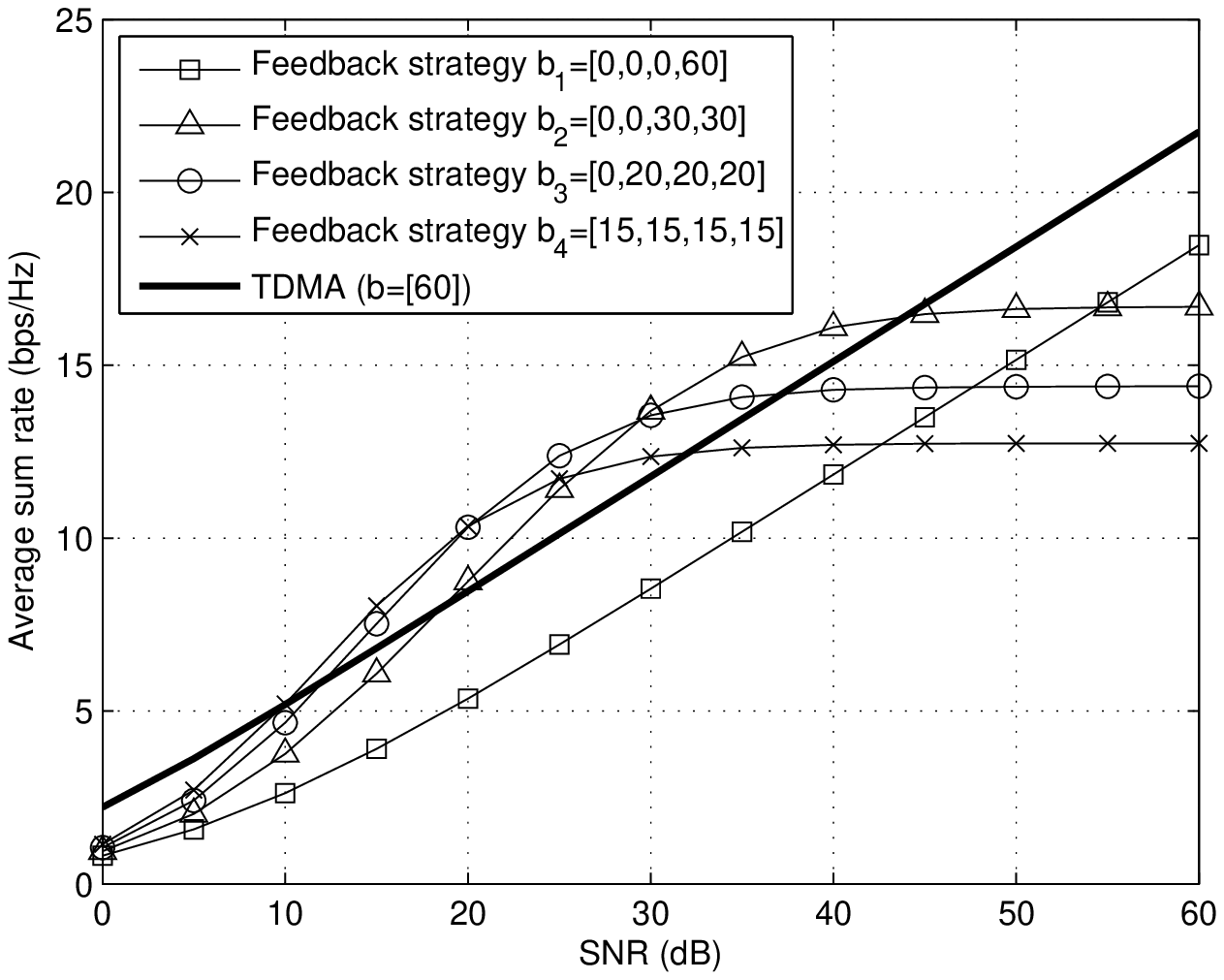}
    \label{fig:sum_rate_TDMA}}\\
\subfigure[Regularized ZF beamforming vs. TDMA]
    {\includegraphics[width=.51\columnwidth]{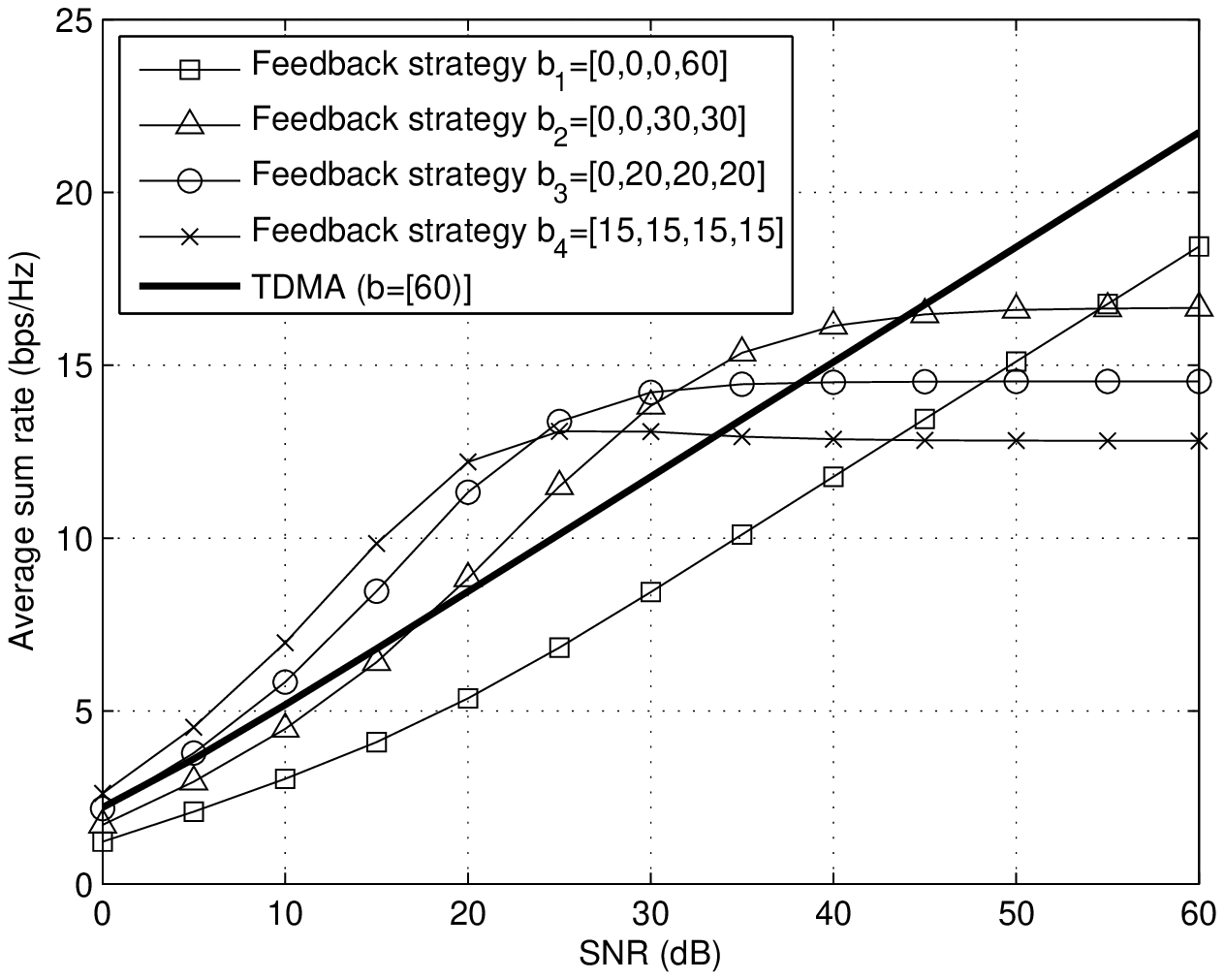}
    \label{fig:RVQ_TDMA_MMSE_60bits}}
\caption{Sum rates of a $4\times 4$ MIMO BC using various feedback
rate sharing strategies ($\sum b_k=60$).} \label{fig:RVQ_TDMA_MMSE}
\end{figure}

\end{document}